\documentclass[11pt]{article}

\usepackage[a4paper, left=0.7in, right=0.7in, top=1.1in, bottom=1.1in]{geometry}
\usepackage[linesnumbered, ruled, vlined]{algorithm2e}
\SetKwInput{KwInput}{Input}                %
\SetKwInput{KwOutput}{Output}              %

\usepackage{mathtools}
\mathtoolsset{showonlyrefs,showmanualtags}
\usepackage{natbib}
\usepackage{amsmath}
\usepackage{amsthm}
\usepackage{amssymb}
\usepackage{amsfonts}
\usepackage{bm}
\usepackage{physics}
\usepackage{multirow}
\usepackage{booktabs}
\usepackage{url}
\usepackage{longtable}
\usepackage{textcomp}
\usepackage{enumerate}
\newtheorem{theorem}{Theorem}[section]

\newtheorem{corollary}[theorem]{Corollary}

\newtheorem{proposition}[theorem]{Proposition}
\newtheorem{definition}[theorem]{Definition}

\newtheorem{example}[theorem]{Example}
\newtheorem{remark}[theorem]{Remark}

\newcommand{\supp}[1]{\mathrm{supp}({#1})}
\newcommand{\lcm}[1]{\mathrm{lcm}({#1})}
\newcommand{\XX}{\mathbb{X}}
\newcommand{\TT}{\mathcal{T}}
\newcommand{\RR}{\mathbb{R}}
\newcommand{\PP}{\mathcal{P}}
\newcommand{\II}{\mathcal{I}}

\newcommand{\degree}[1]{\mathrm{deg}({#1})}

\newcommand{\degreek}[2]{\mathrm{deg}_{{#1}}({#2})}

\newcommand{\cnorm}[1]{\norm{{#1}}_{\mathrm{c}}}
\newcommand{\enorm}[1]{\norm{{#1}}_{2}}
\newcommand{\gwnorm}[1]{\norm{{#1}}_{\mathrm{\nabla},\mathbb{X}}}

\newcommand{\mc}[1]{\mathcal{{#1}}}

\usepackage{hyperref}
\hypersetup{
    colorlinks = true,
    linkcolor = {blue},
    citecolor = {blue},
    urlcolor = {blue},
}
\usepackage{authblk}  
\usepackage{multirow}

\begin{document}

\title{Gradient-Weighted, Data-Driven Normalization \\ for Approximate Border Bases---Concept and Computation}

\author[1,2]{Hiroshi Kera\thanks{Email: \href{mailto:kera@chiba-u.jp}{kera@chiba-u.jp}}}
\author[3]{Achim Kehrein\thanks{Email: \href{achim.kehrein@hochschule-rhein-waal.de}{achim.kehrein@hochschule-rhein-waal.de}}}

\affil[1]{Chiba University}
\affil[2]{Zuse Institute Berlin}
\affil[3]{Rhine-Waal University of Applied Sciences}

\date{}  %

\maketitle

\begin{abstract}
This paper studies the concept and the computation of approximately vanishing ideals of a finite set of data points. By data points, we mean that the points contain some uncertainty, which is a key motivation for the approximate treatment. A careful review of the existing border basis concept for an exact treatment motivates a new adaptation of the border basis concept for an approximate treatment. In the study of approximately vanishing polynomials, the normalization of polynomials plays a vital role. So far, the most common normalization in computational commutative algebra uses the coefficient norm of a polynomial. Inspired by recent developments in machine learning, the present paper proposes and studies the use of gradient-weighted normalization. The gradient-weighted semi-norm evaluates the gradient of a polynomial at the data points. This data-driven nature of gradient-weighted normalization produces, on the one hand, better stability against perturbation and, on the other hand, very significantly, invariance of border bases with respect to scaling the data points. Neither property is achieved with coefficient normalization. In particular, we present an example of the lack of scaling invariance with respect to coefficient normalization, which can cause an approximate border basis computation to fail. This is extremely relevant because scaling of the point set is often recommended for preprocessing the data. Further, we use an existing algorithm with coefficient normalization to show that it is easily adapted to gradient-weighted normalization. The analysis of the adapted algorithm only requires tiny changes, and the time complexity remains the same. 
Finally, we present numerical experiments on three affine varieties to demonstrate the superior stability of our data-driven normalization over coefficient normalization. We obtain robustness to perturbations and invariance to scaling.
\end{abstract}

\section{Introduction}
\label{sec:Introduction}

A finite set of points $\XX=\{p_1,\dots,p_m\}\subset \RR^n$ can be described by its vanishing ideal of polynomials,
\begin{align}
    \mathcal{I}(\XX)=\{f\in \RR[x_1,\dots,x_n]\mid f(p_1)=\dots=f(p_m)=0\}\ .
\end{align}
There are effective algorithms that compute a finite set of generating polynomials, $\langle g_1,\dots,g_\nu\rangle=\mathcal{I}(\XX)$, for example, a Gr\"obner basis or a border basis of the vanishing ideal \citep{cox1992ideals, kreuzer2005computational, dickenstein2005solving}. If the point coordinates are exactly known rational numbers, these computations can be performed exactly. However, recent applications use points that are only known with uncertainty. Examples include machine learning~\citep{livni2013vanishing,kiraly2014dual,hou2016discriminative,kera2016vanishing,kera2018approximate,wirth2022conditional,wirth2023approximate,pelleriti2025latent}, computer vision~\citep{zhao2014hand,yan2018deep}, robotics~\citep{iraji2017principal,antonova2020analytic}, nonlinear systems~\citep{torrente2009application,kera2016noise,karimov2020algebraic,karimov2023identifying}, and signal processing~\citep{wang2018nonlinear, wang2019polynomial}. The notion of the vanishing ideal has to be extended from the exact setting to an approximate setting. This extension starts with the concept of vanishing for polynomials.

\subsection{Approximately Vanishing Polynomials -- Geometric Distance and Algebraic Distance}
\label{subsec:Intro-Approximately-Vanishing-Polynomials}

For a finite set $\XX$ of points that are considered only known with uncertainty, exactly vanishing polynomials become irrelevant. For instance, if data points are almost collinear, it may make more sense to describe the data set by linear polynomials. Alhough these linear polynomials may not vanish at the data points, they vanish at points \textit{close} to the data points. This notion of proximity considers the distance in $\RR^n$ between a data point with uncertainty and a point where the polynomial vanishes. We call this distance the \textit{geometric distance}. 

However, this geometric distance is a very inconvenient concept for calculations, and we replace it with the evaluation of a polynomial $f$ at the given data point set $\XX$. Heuristically, a polynomial that vanishes on points close to the data points must have ``small'' evaluations at the data points themselves. Instead of requiring that polynomials vanish at points close to the data points, we demand that polynomials almost vanish at the given data points. Therefore, we turn our attention to $\epsilon$-vanishing polynomials $f$ on $\XX$,
\begin{align}\label{eq:tolerance}
\enorm{f(\XX)}=\sqrt{f(p_1)^2+\dots+f(p_m)^2}\le \epsilon
\end{align}
for a prescribed tolerance $\epsilon\ge 0$. To distinguish this function-value approach from the argument-approach of the geometric distance, we call this the \textit{algebraic distance}.

The tolerance condition \eqref{eq:tolerance} by itself is insufficient. For any non-exactly vanishing polynomial on $\XX$, we can decrease or increase its evaluation by scalar multiplication: while $\gamma\cdot f$  is $\epsilon$-vanishing for $\gamma\le\epsilon/\enorm{f(\XX)}$, it is not $\epsilon$-vanishing for $\gamma >\epsilon/\enorm{f(\XX)}$. In \citep{kera2019spurious}, this problem of scaling arbitrariness is named the \textit{spurious vanishing problem}: every polynomial---however geometrically unrelated its zeros are to the given point set---is approximately vanishing if the polynomial is multiplied by a sufficiently small scalar. 

Non-zero constant polynomials $f=c\ne 0$ underline this scaling problem. Their evaluation on  $\XX$ is
\begin{align}
    \enorm{f(\XX)}=\sqrt{|\XX|\cdot c^2}=|c|\cdot\sqrt{|\XX|}\ .
\end{align}
Without a normalization constraint, constant polynomials $c\ne 0$, but sufficiently close to zero, would be $\epsilon$-vanishing. Certainly, this is geometrically unintended. Non-zero constant polynomials should never be considered as vanishing on a non-empty set of points. Hence, we treat non-zero constant polynomials as exceptional cases.

To avoid the scaling influence for non-constant polynomials, we are going to combine the tolerance condition~\eqref{eq:tolerance} with a normalization constraint, $\|f\|=1$, for some semi-norm $\norm{\,\cdot\,}$ on polynomials. A suitable choice of this semi-norm is a central aspect of this paper.

\subsection{Gradient-Weighted and Data-Driven Normalization}
\label{subsec:Intro-Normalization}

A common normalization constraint in computational commutative algebra uses the coefficient norm in the polynomial ring $\mc{P}=\RR[x_1,\dots,x_n]$,
\begin{align}
    \cnorm{f}=\cnorm{\sum\nolimits_{\, t\in\TT} c_t\cdot t}=\sqrt{\sum\nolimits_{\,t\in \TT} c_t^2}\ ,
\end{align}
where $\TT$ denotes the set of terms in $\PP$, i.e., the products of indeterminates. 

However, the choice of the coefficient norm for normalization is rather arbitrary. This choice ignores the prescribed point set $\XX$. Moreover, polynomials react differently to perturbations of the data points. To remedy these shortcomings, we adopt ideas from machine learning and propose a normalization that evaluates the gradients of terms at the data points: our \textit{gradient-weighted semi-norm}---to be defined in Section~\ref{sec:gradient-weighted-semi-norm}---is a hybrid of coefficient normalization from computational commutative algebra and of gradient normalization from machine learning. A coefficient $c_t$ is weighted by the gradient of its term $t$ evaluated at the prescribed point set $\XX$. The use of $\XX$ makes the gradient-weighted semi-norm \textit{data-driven}. 

The gradient relates the algebraic with the geometric distance. Ideally, for each potentially perturbed point $p\in \XX$, an approximate vanishing polynomial $f$ exactly vanishes for some nearby point $p^*$. Let $\XX^*=\{p^*_1,\dots,p^*_m\}$ be the set of such \textit{clean} points and let $\Delta \XX=(p_1-p_1^*,\dots,p_m-p_m^*)$ be the vector of perturbations, which contains the information about the geometric distance. Then 
\begin{align}
    \underbrace{\enorm{f(\XX)}}_{\text{algebraic distance}} = \ 
    \|f(\XX)-\underbrace{f(\XX^*)}_{= 0}\|_2\ \approx\ \enorm{f(\XX)-\bigl(f(\XX)-\nabla f(\XX)\cdot \Delta\XX\bigr)}\ \le  \ \underbrace{\enorm{\nabla f(\XX)}}_{\text{evaluated gradient}}\cdot \underbrace{\enorm{\Delta\XX}}_{\text{geometric distance}}\ .
\end{align}
Therefore, the norm of the gradient of the polynomial evaluated at the data points is an important heuristic of how well algebraic and geometric distance are related.

\subsection{Organization of the Paper}
\label{subsec:Intro-Organization-Paper}

After this introduction, Section~\ref{sec:related-work} briefly reviews related work on approximately vanishing polynomials. We focus on the computation of approximate border bases~\citep{abbott2008stable,heldt2009approximate,limbeck2013computation}, since they have greater numerical stability than Gr\"obner bases~\citep{stetter2004numerical,fassino2010almost}. 

Section~\ref{sec:basic-notation} begins with the definitions, discusses some details of approximately vanishing ideals, and adapts the concept of a border basis to our approximation and normalization needs. Then, Section~\ref{sec:gradient-weighted-semi-norm} presents the key concept of this paper: the gradient-weighted semi-norm. Gradient-weighted normalization can easily be implemented in many algorithms for computing bases of vanishing ideals: replace eigenvalue problems with generalized eigenvalue problems. This is such a minor modification that the analysis of the modified algorithms hardly differs. Section~\ref{sec:implementing-gradient-normalization}, for concreteness, implements gradient-weighted normalization in the approximate Buchberger--M\"oller (ABM) algorithm, which has been studied by Limbeck with coefficient normalization~\citep{limbeck2013computation}. 

Gradient-weighted normalization features (i) robustness against perturbation of the data points and (ii) invariance of the scale of the points. This is shown in Section~\ref{sec:gradient-weighted-features}. Border basis polynomials computed from two scales of the data points are also just scaled versions of one another. By contrast, Section~\ref{sec:coefficient-scaling-dependence} shows that approximate bases with respect to coefficient normalization need not be scale-invariant. There exists a point set with two scales, for which the respectively computed basis polynomials are not scaled versions of one another. In particular, the customary preprocessing of scaling the data set into the $n$-dimensional cube $[-1,1]^n$ turns out to be problematic with respect to coefficient normalization. 

Finally, Section~\ref{sec:numerical-experiments} documents the results of some numerical experiments. Three affine varieties are sampled finitely many times with perturbation. We compute bases of approximately vanishing ideals of the finite sample sets for various tolerances and study how well these bases approximate the exact vanishing ideal of the original algebraic variety.

\section{Related Work in Machine Learning and Computational Commutative Algebra}
\label{sec:related-work}

Approximately vanishing polynomials appear in various fields such as reconstruction of dynamical systems~\citep{torrente2009application,kera2016noise,karimov2020algebraic}, signal processing~\citep{wang2018nonlinear,wang2019polynomial}, and machine learning~\citep{hou2016discriminative,kera2016vanishing,antonova2020analytic}. Accordingly, the computation of bases of vanishing ideals with approximately vanishing polynomials has been studied extensively during the last decades~\citep{abbott2008stable,heldt2009approximate,fassino2010almost,robbiano2010approximate,limbeck2013computation,livni2013vanishing,kiraly2014dual,kera2018approximate,kera2019spurious,kera2020gradient,wirth2022conditional,kera2022border,wirth2023approximate,pelleriti2025latent}. 

In machine learning, the basis computation of vanishing ideals is performed in a term-agnostic manner to circumvent symbolic computation and term orderings~\citep{livni2013vanishing,kiraly2014dual}. In this approach, there is no efficient access to the coefficients of terms, and polynomials are often handled without proper normalization. The resulting problem of spuriously vanishing polynomials was pointed out by  
\cite{kera2019spurious}. Then, \cite{kera2020gradient,kera2024monomial} addressed the lack of normalization by introducing the gradient semi-norm $(\sum_{p \in \XX} \enorm{\nabla g(p)}^2)^{1/2}$. Moreover, this semi-norm specifically incorporates data-driven aspects. 

By contrast, the term-aware basis computations in computational commutative algebra used normalization, usually with respect to the coefficient norm. As it seemed at first that the use of a gradient semi-norm would merely increase computational cost, algebraic approaches have often ignored the data-driven aspect of this new normalization. 

The gradient has been exploited for approximate computations of vanishing ideals in some studies. For example, \cite{abbott2008stable} computed the first-order approximation of polynomials, i.e.,  their gradient,  to determine a set of terms whose evaluation matrix maintains full-rank for small perturbations in the points. This approach incurs heavy computational costs and strong sensitivity to an additional parameter. Similarly, ~\cite{fassino2013simple} consider the first-order approximation of polynomials to compute a low-degree polynomial that vanishes on points geometrically close to the given points. It focuses on the lowest degree polynomial and does not compute a basis. Nevertheless, both approaches use coefficient normalization.

One of the earliest occurrences of data-driven normalization of approximately vanishing polynomials appears in \citep{kera2020gradient,kera2024monomial}. Similarly to the present paper, they normalize polynomials with respect to a gradient semi-norm. However, their method focuses on term-agnostic computations, which may be helpful when symbolic computation and term orderings are contraindicated. For example, several data-driven applications including machine learning exploit fast numerical computation libraries and powerful hardware accelerators such as graphics processing units (GPU). However, to our knowledge, there is no framework that seamlessly runs numerical and symbolic computations on GPUs. Besides, in such applications, one often needs to interpret the intrinsic structure of data from the output polynomials. The term ordering is problematic for this use as it introduces an artificial bias to the outputs. 

It has been unknown how helpful data-driven normalization is in the term-\textit{aware} setting---the standard in computer algebra. Consequentially, the relationship between the gradient semi-norm and the coefficient norm has not been studied. While the implementation of gradient normalization in border basis computations cannot exploit all advantages of term-agnostic basis computations, the gradient-weighted normalization presented in this paper transfers many nice properties of gradient normalization from term-agnostic computations to term-aware computations. The algorithmic complexity stays in the same order of magnitude.

\section{Basic Concepts and Definitions: Review and Adaptation}\label{sec:basic-notation}

Let $\PP=\RR[x_1,\ldots, x_n]$ denote the ring of polynomials in the indeterminates $x_1,\ldots, x_n$ with real coefficients. The set of all \textbf{terms} $t=x_1^{e_1}\cdots x_n^{e_n}$ with integer exponents $e_1,\dots,e_n\ge 0$ is labeled $\TT$. For $1\le k\le n$, the $k$-th \textbf{partial degree of the term} $t$ is $\degreek{k}{t}=e_k$ and its \textbf{(total) degree} is $\degree{t}=\sum_{k=1}^n e_k$.

A polynomial is a finite linear combination of terms, $\sum_{t\in \TT} c_t t$. All sums in this paper are finite. The \textbf{support of a polynomial} $f=\sum_{t\in \TT} c_t t$ is the finite set of terms that occur with a nonzero coefficient, $\supp{f}=\{t\in\TT\mid c_t\ne 0\}$. The $k$-th \textbf{partial degree of the polynomial} $f$ is $\degreek{k}{f}=\max\{\degreek{k}{t}\mid t\in \supp{f}\}$, while its \textbf{degree} is $\degree{f}=\max\{\deg{t}\mid t\in \mathrm{supp}(f)\}$. 

The \textbf{coefficient norm} of a polynomial is 
\begin{align}
    \cnorm{f}= \Bigl(\sum_{t\in \supp{f}} |c_t|^2\Bigr)^{1/2}\ .
\end{align}
A polynomial will be \textbf{coefficient-normalized}, if $\cnorm{f}=1$.

\subsection{Evaluation of Polynomials and Approximately Vanishing Polynomials}\label{subsec:approximately-vanishing-polynomials}

Let $\XX = \{p_1,\dots,p_m\} \subset \RR^n$ be a finite, non-empty set of points. We often consider its elements enumerated to fix an order in definitions and algorithms. More formally, we may think of $\XX$ as a tuple of non-repeating points. The \textbf{evaluation on} $\XX$ assigns to each polynomial the vector of its function values,
\begin{align}
    \text{eval}_{\XX}\colon \RR[x_1,\dots, x_n]\to \RR^m,\quad f\mapsto f(\XX)=
    \begin{pmatrix}
    f(p_1)\\
    \vdots\\
    f(p_m)
    \end{pmatrix}\ .
\end{align}
We call $f(\XX)$ the \textbf{evaluation vector} and use pointwise multiplication to make the vector space $\RR^m$ a commutative algebra with unit $(1,\dots,1)^\intercal$ and zero divisors.
Then the evaluation map $\text{eval}_{\XX}$ is a ring homomorphism with $\text{eval}_{\XX}(1)=(1,\dots,1)^\intercal$ and its kernel is the vanishing ideal $\mathcal{I}=\mathcal{I}(\XX)$. The evaluation map induces the isomorphism
\begin{align}
    \PP/\mathcal{I}=\RR[x_1,\dots,x_n]/\mathcal{I}\cong \RR^m\ .
\end{align}
We are particularly interested in vector space bases of $\PP/\mathcal{I}$ induced by terms.

For a tuple of polynomials, we extend the evaluation vector to an \textbf{evaluation matrix},
\begin{align}
    \begin{pmatrix}
    f_1(\XX)& \dots &f_r(\XX)
    \end{pmatrix}
    =\begin{pmatrix}
    f_1(p_1) & \dots & f_r(p_1)\\
    \vdots & &\vdots \\
    f_1(p_m) & \dots & f_r(p_m)\\
    \end{pmatrix}
    \in \RR^{m\times r}\ .
\end{align}
Further, we evaluate the gradient of $f$ on $\XX$ by concatenating the column vector gradients $\nabla f(p_1)^\intercal$, \dots, $\nabla f(p_m)^\intercal$ into a long column vector. Analogously, gradient evaluations of tuples of polynomials are matrices.
\begin{align}
    \nabla f(\XX)=\begin{pmatrix}
    | \\
    \nabla f(p_1)^\intercal\\
    |\\
    \vdots\\
    | \\
    \nabla f(p_m)^\intercal\\
    |\\
    \end{pmatrix}
    \quad\text{and}\quad
    \begin{pmatrix}
    \nabla f_1(\XX) & \dots & \nabla f_r(\XX)
    \end{pmatrix}=\begin{pmatrix}
    |  &  & | \\
    \nabla f_1(p_1)^\intercal & \dots &\nabla f_r(p_1)^\intercal\\
    |&  & | \\
    \vdots &&\vdots \\
    |&  & | \\
    \nabla f_1(p_m)^\intercal & \dots &\nabla f_r(p_m)^\intercal\\
    |& & | \\
    \end{pmatrix}
    \ .
\end{align}

The following is our key definition.

\begin{definition}
Let $\XX=\{p_1,\dots,p_m\}\subset \RR^n$ be a finite, non-empty set of points and $\epsilon\ge 0$ a tolerance. A polynomial $f\in \PP$ is \textbf{$\boldsymbol{\epsilon}$-vanishing on} $\boldsymbol{\XX}$ if the Euclidean norm of its evaluation vector is within the tolerance,
\begin{align}
    \enorm{f(\XX)}=\left(\sum_{j=1}^m |f(p_j)|^2\right)^{1/2}\le\epsilon\ .
\end{align}
A zero-vanishing polynomial on $\XX$ is simply called \textbf{vanishing} on $\XX$ or---for emphasis---\textbf{exactly vanishing} on $\XX$.
\end{definition}

As outlined in Section~\ref{sec:Introduction}, scalar multiplication affects this concept of approximately vanishing. On the one hand, for every positive tolerance $\epsilon >0$, every polynomial $f$ can be scaled down to become $\epsilon$-vanishing. If $f$ is exactly vanishing, nothing needs to be done. If $f$ is not exactly vanishing on $\XX$, then $\enorm{\alpha\cdot f(\XX)}\le \epsilon$ for any $\alpha\le \epsilon/\enorm{f(\XX)}$. On the other hand, every non-vanishing polynomial can be scaled up to fail the $\epsilon$-vanishing condition. For $\alpha >\epsilon/\enorm{f(\XX)}$ the scalar multiple $\alpha \cdot f$ is not $\epsilon$-vanishing. 

To avoid this scaling dependence, $\epsilon$-vanishing conditions will next be combined with normalization constraints on the polynomials.

\subsection{Approximately Vanishing Ideals}\label{subsec:approximate-ideal}

The following definition generalizes a definition in~\citep{heldt2009approximate} from coefficient-normalization to other normalizations.

\begin{definition} \label{def:approximate-vanishing-ideal}
Let $\XX\subset \RR^n$ be a non-empty finite set of points and $\epsilon\ge 0$ a tolerance. Let $\norm{\,\cdot\,}$ be a semi-norm on the algebra of polynomials $\PP$. 
An ideal $\mathcal{I}\subset \PP$ is \textbf{$\epsilon$-vanishing on $\XX$ with respect to $\norm{\,\cdot\,}$-normalization}, if there exist polynomials $g_1,\dots, g_r$ that generate $\mathcal{I}$ such that each $g_i$ either vanishes exactly on $\XX$ (no normalization required) or is $\norm{\,\cdot\,}$-normalized and $\epsilon$-vanishes on $\XX$.
\end{definition}

The relationship between the vanishing ideal $\II(\XX)$ and an $\epsilon$-vanishing ideal $\II$ on $\XX$ is complicated. First of all, unlike the vanishing ideal $\II(\XX)$, an $\epsilon$-vanishing ideal on $\XX$ is not uniquely determined. For example, any subset of the generators of an $\epsilon$-vanishing ideal on $\XX$ is also a set of generators for a possibly different $\epsilon$-vanishing ideal on $\XX$. In particular, if we use a subset of exactly vanishing generators, even a zero-vanishing ideal may be strictly larger than $\II(\XX)$. Secondly, an $\epsilon$-vanishing generator $g_j$ need not be exactly vanishing, so an approximately vanishing ideal $\mathcal{I}$ need not be contained in $\II(\XX)$. Conversely, a non-zero exactly vanishing polynomial on $\XX$ need not be in the $\epsilon$-vanishing ideal $\II$. The following example illustrates some of these statements.  

\begin{example} \label{ex:epsilon-vanishing-ideal}
Consider $\mathbb{X}=\{(0.9, 1.1), (1.1, 0.9)\}$ and coefficient normalization of polynomials. The following examples are illustrated in Figure~\ref{fig:vanishing-ideals}.

\begin{figure} [h]
    \centering
    \includegraphics[width=0.95\textwidth]{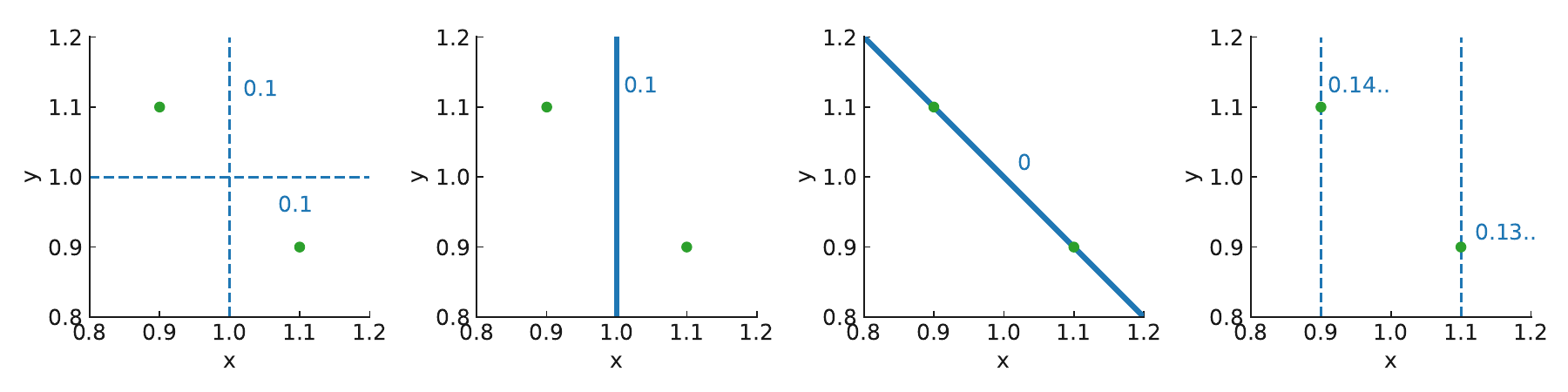}
    \caption{Examples of the polymorphy of approximately vanishing ideals on a set of two points, see Example~\ref{ex:epsilon-vanishing-ideal}. The two points $(0.9,1.1)$ and $(1.1,0.9)$ are shown in red. The zero sets of a single generating polynomial are shown as a blue line---solid, if it belongs to all generators, and dashed, if it does not belong to all generators. The numbers indicate the Euclidean norm of the evaluation of the generator on $\XX$, i.e.,\ the lower bound for which the coefficient-normalized generator becomes approximately vanishing on $\XX$. From left to right: (i) The generators span the vanishing ideal of the blue point, which is the intersection of the zero sets. For a tolerance $\epsilon\ge 0.1$, both generators become $\epsilon$-vanishing. (ii) The single generator spans the exact-vanishing ideal of the vertical line. This ideal is $\epsilon$-vanishing for $\epsilon\ge 0.1$ (iii) The single generator spans the exact-vanishing ideal of the slanted line. This ideal is exact-vanishing on $\XX$. (iv) The two generators have no zeros in common. Their vanishing ideal is the unit ideal, which becomes $\epsilon$-vanishing for $\epsilon \ge 0.14\dots$, which is the larger algebraic distance of the two generators. The evaluations of the polynomials are positive since one point is not a zero. The evaluations differ, since the coefficient normalizations differ.}
    \label{fig:vanishing-ideals}
\end{figure}

\begin{enumerate}[(i)]
\item The coefficient-normalized polynomials 
    $h_1=(x-1)/\sqrt{2}$ and $h_2=(y-1)/\sqrt{2}$
generate the vanishing ideal $\mathcal{I}(\mathbb{Y})$ of the single point $\mathbb{Y}=\{(1,1)\}$. We also have
\begin{align}
    h_1(\mathbb{X})=\frac{1}{\sqrt{2}}\begin{pmatrix}
    -0.1 \\ \phantom{-}0.1
    \end{pmatrix}
    \quad\text{and}\quad 
    h_2(\mathbb{X})=\frac{1}{\sqrt{2}}\begin{pmatrix}
    \phantom{-}0.1 \\ -0.1
    \end{pmatrix}
    \quad\text{with}\quad
    \|h_1(\mathbb{X})\|_2=\|h_2(\mathbb{X})\|_2=\sqrt{\frac{0.02}{2}}=0.1\ .
\end{align}
So, for any tolerance $\epsilon \ge 0.1$ the generators $h_1$ and $h_2$ are $\epsilon$-vanishing on $\XX$, coefficient-normalized, and $\langle h_1,h_2\rangle=\II(\mathbb{Y})$ is an $\epsilon$-vanishing ideal on $\XX$. Neither of the vanishing ideals $\II(\mathbb{Y})$ and $\II(\XX)$ is contained in the other. The points of $\XX$ can be seen as two differently perturbed samples of the point in $\mathbb{Y}$.

\item The same polynomial $h_1=(x-1)/\sqrt{2}$ by itself generates the vanishing ideal $I(\mathbb{L})$ of the line $\mathbb{L}=\{(1,y)\mid y\in\RR\}$. For any tolerance $\epsilon \ge 0.1$ the generator $h_1$ is $\epsilon$-vanishing on $\XX$, coefficient-normalized, and $\langle h_1\rangle$ is an $\epsilon$-vanishing ideal on $\XX$. Again, neither of the vanishing ideals $\II(\mathbb{L})$ and $\II(\XX)$ is contained in the other. With regard to the previous example, we have $\II(\mathbb{L})\subsetneq \II(\mathbb{Y})$. Again, the points of $\XX$ can be seen as two perturbed samples of points in $\mathbb{L}$.

\item The polynomial $f=(x+y-2)/\sqrt{6}$ is coefficient normalized and vanishes exactly on $\XX$. It generates the vanishing ideal $\II(\mathbb{U})$ of the line $y=2-x$, which is strictly contained in the vanishing ideal $\II(\XX)$. The points of $\XX$ can be seen as two unperturbed samples of points in $\mathbb{U}$.

\item The polynomials $f_1=(x-1.1)/\sqrt{2.21}$ and $f_2=(x-0.9)/\sqrt{1.81}$ are coefficient-normalized with evaluations
\begin{align}
    \enorm{f_1(\XX)}=\frac{1}{\sqrt{2.21}}\enorm{\begin{pmatrix}
        -0.2 \\ 0
    \end{pmatrix}}=\frac{0.2}{\sqrt{2.21}}= 0.134\dots
    \quad\text{and}\quad
    \enorm{f_2(\XX)}=\frac{1}{\sqrt{1.81}}\enorm{\begin{pmatrix}
        0 \\ 0.2
    \end{pmatrix}}=\frac{0.2}{\sqrt{1.81}}=0.148\dots \ .
    \end{align}
    They generate an $0.15$-vanishing ideal $\II$ on $\XX$. Since $\sqrt{2.21}\cdot f_1-\sqrt{1.81}\cdot f_2=-0.2$, this ideal contains constant polynomials and equals the unit ideal. Nevertheless, a constant polynomial does not qualify as a generator of a vanishing ideal of a non-empty set. The coefficient normalized constant polynomial is $1$, which is not $0.15$-vanishing, since $\enorm{1(\XX)}=\sqrt{2}$. The points of $\XX$ can be seen as two samples from the parallel lines.
\end{enumerate}
\end{example}

Approximately vanishing ideals of $\XX$ can be either the unit ideal or a vanishing ideal of a not necessarily finite set, of which the points of $\XX$ can be considered finitely many, possibly perturbed samples. There are many more examples of approximately vanishing ideals on this set of two points. In particular, the generators can be chosen to be nonlinear polynomials.

Caution, being an $\epsilon$-vanishing ideal with respect to $\|\cdot\|$-normalization does not mean that each normalized polynomial is $\epsilon$-vanishing. For example, if the $\epsilon$-vanishing ideal is the unit ideal, the coefficient-normalized term 1 is not $\epsilon$-vanishing for sufficiently small $\epsilon$. The existence of an $\epsilon$-vanishing, $\|\cdot\|$-normalized generating set does not force the vanishing property on other normalized polynomials.

Next, we look at bases, i.e., generating sets of ideals that possess a certain structure.

\subsection{Exact and Approximate Border Bases}\label{subsec:border-bases}

We recall the definitions of an order ideal and an exact border basis, which will then be extended to our approximate setup.

\begin{definition}[\cite{kehrein2005algebraist}] \label{def:order-border-neighbor}
\begin{enumerate}[(i)]
\item A finite set of terms $\mathcal{O}\subset \TT$ is an \textbf{order ideal}, if the set is closed under divisors: if $t\in\mathcal{O}$ and $s\in \TT$ divides $t$, then $s\in \mathcal{O}$. 
\item The \textbf{border} of $\mathcal{O}$ is the set of all terms that are products of an order ideal term $t$ and an indeterminate $x_k$, and different from order terms, i.e., $\partial\mathcal{O} = \qty(\bigcup_{k=1}^n x_k\mathcal{O}) \setminus \mathcal{O}$. 
\item Two border terms $b_i, b_j \in \partial\mathcal{O}$ are \textbf{neighbors} if one is an indeterminate-multiple of the other, $b_i = x_kb_j$ or $b_j = x_kb_i$ (\textbf{next-door neighbors}), or, if they have a common indeterminate-multiple, $x_kb_i = x_\ell b_j$ (\textbf{across-the-street neighbors}). 
\end{enumerate}
\end{definition}

Given a polynomial ideal, we are interested in a finite set of generating polynomials that are supported by an order ideal and its border: each border term corresponds to a generator whose support consists of this border term and a subset of the order ideal.

\begin{definition}[\cite{kehrein2005algebraist}; restricted to {$\PP = \RR[x_1,\ldots, x_n]$}]\label{def:border-basis}
Let $\mathcal{O}=\{t_1,\dots,t_s\}$ be an order ideal with border $\partial\mathcal{O}=\{b_1,\dots,b_r\}$. Let $\mathcal{I}\subset \PP$ be a polynomial ideal. An \textbf{$\mathcal{O}$-border prebasis} $G=\{g_1,\dots,g_r\}\subset\mathcal{I}$ is a set of polynomials of the form 
\begin{align}
    g_i=c_i\, b_i - \sum_{1\le j\le s} c_{ij}\, t_j\ ,\quad c_i>0,\ c_{ij}\in \RR,\quad 1\le i\le r,\  1\le j\le s\ .
\end{align}
An alternative notation uses indexing by terms,
\begin{align}
    g_b=c_b\, b - \sum_{t\in\mathcal{O}} c_{b,t}\, t\ ,\quad c_b>0,\ c_{b,t}\in \RR,\ b\in\partial\mathcal{O}\ .
\end{align}
If the cosets $t_1+\mathcal{I},\dots, t_s+\mathcal{I}$ of the order ideal terms form a basis of the $\mathbb{R}$-vector space $\PP/\mathcal{I}$, then $G$ is an \textbf{$\mathcal{O}$-border basis} of $\mathcal{I}$. %
\end{definition}

In the commonly used algebraic definition, the border polynomials are ``normalized'' to a unit border coefficient $c_i=1$, which is nice for elimination. Since this is incompatible with our normalization needs, we have adapted the definition and allow positive border coefficients $c_i>0$. By excluding negative border coefficients, a normalized border polynomial is uniquely determined.

A helpful characterization is that a prebasis in $\mathcal{I}$ is a basis if and only if $\mathcal{I}\cap \langle \mathcal{O}\rangle_{\RR}=\{0\}$. In other words, in the basis case, the only polynomial in $\mathcal{I}$ which is supported by $\mathcal{O}$ is the zero polynomial. For this characterization and more, see~\citep{kehrein2005charactorizations}.

The defining property of a border basis about the quotient space $\PP/\mathcal{I}$ does not transfer nicely to approximately vanishing ideals. The main problem is that the algebraic structure of an ideal is closed under scalar multiples, so every approximately vanishing polynomial is accompanied by its non-vanishing multiples. Moreover, there is often a linear combination of the generators which produces the constant polynomial 1. The approximately vanishing ideal is often the unit ideal, $\PP=\mathcal{I}$, and $\PP/\mathcal{I}=\{0\}$. Therefore, we use a new definition that focuses on normalized polynomials.

\begin{definition}\label{def:our-approximate-border-basis}
    Let $\XX\subset\RR^n$ be a finite, non-empty set of points and $\epsilon\ge 0$ a tolerance. Let $\|\cdot\|$ be a semi-norm on $\PP$.
    An \textbf{$\epsilon$-approximate, $\|\cdot\|$-normalized $\mathcal{O}$-border basis for $\XX$} is an $\mathcal{O}$-border prebasis $G=\{g_1,\dots, g_r\}$ such that 
    \begin{enumerate}[(i)]
    \item each $g_i$ is exactly vanishing on $\XX$ or can be $\|\cdot\|$-normalized and then is $\epsilon$-vanishing on $\XX$, 
    \begin{align}
        \frac{1}{\|g_i\|}\cdot \enorm{g_i(\XX)}\le \epsilon\ ,
    \end{align}
    and
    \item all non-constant polynomials $f$ supported by $\mathcal{O}$ can be $\|\cdot\|$-normalized and then are not $\epsilon$-vanishing,
    \begin{equation*}
        \frac{1}{\|f\|}\cdot \enorm{f(\XX)}>\epsilon\ .
    \end{equation*}
    \end{enumerate}
\end{definition}
Condition (i) makes $\mathcal{I}=\langle g_1,\dots,g_r\rangle$ an $\epsilon$-vanishing ideal of $\XX$ by Definition~\ref{def:approximate-vanishing-ideal}. Condition (ii) transfers the linear independence of the term residue classes in $\PP/\mathcal{I}$ from exact border bases to the approximate setting. In the exact case, any non-trivial linear combination of terms in $\mathcal{O}$ must be non-zero in $\PP/\mathcal{I}$, i.e., it must not belong to the vanishing ideal. In the approximate setting, a non-constant, $\|\cdot\|$-normalized polynomial supported by $\mathcal{O}$ must not $\epsilon$-vanish. In particular, every non-constant order term $t\in\mathcal{O}$ must be $\|\cdot\|$-normalizable and must satisfy $\enorm{t(\XX)}/\|t\|>\epsilon$. The remaining case, a non-zero scalar multiple of the constant term, is considered not to vanish, not even approximately. 

Our Definition~\ref{def:our-approximate-border-basis} deviates from the classical one~\citep{heldt2009approximate}. We review next and then address some of the differences. The classical definition is based on the following combinations of border prebasis polynomials that eliminate the original border terms. 

\begin{definition}\label{def:S-polynomial}
Let $\mathcal{O}=\{t_1,\dots,t_s\}$ be an order ideal with border $\partial\mathcal{O}=\{b_1,\dots,b_r\}$ and let $G=\{g_1,\dots,g_r\}$ be an $\mathcal{O}$-border prebasis, 
\begin{align}
    g_i=c_i\, b_i+\sum_{j=1}^s c_{ij}\,t_j\quad\text{for}\ 1\le i\le r\ .
\end{align}
For neighboring border terms $b_i$ and $b_j$ (cf.~Definition~\ref{def:order-border-neighbor}) we consider the \textbf{S-polynomial} (syzygy)
    \begin{align}
    S(g_i,g_j) = \frac{\lcm{b_i,b_j}}{c_i\,b_i}g_i - \frac{\lcm{b_i,b_j}}{c_j\,b_j}g_j\ , 
    \end{align}
    where  $\lcm{b_i,b_j}$ denotes the least common multiple of $b_i$ and $b_j$.
\end{definition}

The least common multiple of two neighboring border terms $b_i$ and $b_j$ is either one of the border terms (next-door neighbors) or the common indeterminate multiple of the border terms (across-the-street neighbors). The S-polynomial is defined so that the multiples of the border terms cancel out. All other occurring terms in the basis polynomials are order ideal terms. In the construction of an S-polynomial, these order terms are multiplied by at most an indeterminate, so the support of the S-polynomial is in $\mathcal{O}\cup \partial \mathcal{O}$. By subtracting suitable scalar multiples of the border polynomials, we eliminate each border term. The result of this elimination is a polynomial supported by $\mathcal{O}$, which is the \textbf{normal remainder} of the S-polynomial. 

For an exact $\mathcal{O}$-border basis, all normal remainders of S-polynomials are zero since they are in $\mathcal{I}\cap \langle\mathcal{O}\rangle_{\RR}=\{0\}$. In other words, a non-zero polynomial $f$ supported by $\mathcal{O}$ must not be in the ideal $\mathcal{I}$. Since we consider an $\epsilon$-vanishing ideal on $\XX$, which may be the unit ideal, the proper transition according to Definition~\ref{def:our-approximate-border-basis} from the exact to the approximate case is: the $\|\cdot\|$-normalized multiple of $f$ must not be $\epsilon$-vanishing on $\XX$. 

Now, let us compare our definition of an approximate border basis with the classical one.

\begin{definition}[\cite{heldt2009approximate}] \label{def:classical-approximate-border-basis}
Let $\delta \ge 0$. An $\mathcal{O}$-border prebasis $G=\{g_1,\dots,g_\nu\}\subset P$ is a \textbf{$\delta$-approximate $\mathcal{O}$-border basis} of an ideal $\mathcal{I}$, if the border polynomials are border-coefficient normalized and if the normal remainders of all S-polynomials have a coefficient norm less than or equal to $\delta$. 
\end{definition}

This definition makes every border prebasis a $\delta$-approximate border basis, as long as $\delta$ is sufficiently large. Since the number of S-polynomials is finite, there is a maximum $\delta_0$ of the coefficient norms of their normal remainders. The border prebasis is $\delta$-approximate if and only if $\delta\ge \delta_0$. Thus, $\delta_0$ measures how much a prebasis fails to be a basis. 

The fundamental difference between the Definitions~\ref{def:our-approximate-border-basis} and~\ref{def:classical-approximate-border-basis} is that the latter does not refer to a point set $\XX$. While the former concept is motivated by approximately describing a point set $\XX$, the latter is about approximating the basis property.
At first glance, the two definitions seem to contradict one another as the former uses an upper bound and the latter a lower bound in connection with polynomials supported by $\mathcal{O}$. However, this seeming contradiction resolves because Definition~\ref{def:classical-approximate-border-basis} does not evaluate the polynomials. The upper bound is on the coefficient norm, i.e., it corresponds to the normalization, not to the evaluation. Definition~\ref{def:our-approximate-border-basis} is data-driven---it takes the point set $\XX$ into account---, while Definition~\ref{def:classical-approximate-border-basis} is data-independently formulated. Most importantly, our Definition~\ref{def:our-approximate-border-basis} reuses the prescribed tolerance $\epsilon$ instead of introducing the second tolerance $\delta$.

The next example shows differences between exactly vanishing and approximately vanishing ideals. It also illustrates the order ideal structure and the border polynomials. 

\begin{example}
Let $\mathbb{X}=\{(0.9, 1.1), (1.1, 0.9)\}$. 

\begin{enumerate}[(i)]
\item An exact border basis for the order ideal $\mathcal{O}=\{1, x\}$ with border $\partial \mathcal{O}=\{y,xy, x^2\}$ is given by
\begin{align}
    g_1=y+x-2,\quad g_2=xy+x^2-2x,\quad\text{and}\quad g_3=x^2-2x+0.99=(x-0.9)(x-1.1)\ .
\end{align}
Since $g_2=x\cdot g_1$ is a redundant generator, we have $\mathcal{I}=\langle g_1,g_3\rangle$. As the polynomials are exactly vanishing, normalization is irrelevant. 
The order ideal terms $1$ and $x$ are of course coefficient-normalized with the linearly independent evaluations
\begin{align}
    1(\mathbb{X})=\begin{pmatrix}
    1 \\ 1
    \end{pmatrix}
    \quad\text{and}\quad
    x(\mathbb{X})=\begin{pmatrix}
    0.9 \\ 1.1
    \end{pmatrix}\ ,
\end{align}
and the Euclidean norms $\enorm{1(\XX)}=\sqrt{2}$ and $\enorm{x(\XX)}=\sqrt{0.81+1.21}=\sqrt{2.02}$. The number of order ideal terms equals the number of prescribed points.

\item Now consider the coefficient-normalized polynomials 
\begin{align}
    h_1=\frac{1}{\sqrt{2}}(x-1)\quad\text{and}\quad h_2=\frac{1}{\sqrt{2}}(y-1)\ .
\end{align}
They generate the exact vanishing ideal $\mathcal{I}(\mathbb{Y})$ of the single point $\mathbb{Y}=\{(1,1)\}$ and constitute a border basis for the order ideal $\mathcal{O}=\{1\}$ with border $\partial\mathcal{O}=\{h_1,h_2\}$. The evaluation vector $1(\mathbb{Y})=(1,1)^{\intercal}$ is nonzero and therefore linearly independent. For the original point set $\XX$ we have
\begin{align}
    h_1(\mathbb{X})=\frac{1}{\sqrt{2}}\begin{pmatrix}
    -0.1 \\ \phantom{-}0.1
    \end{pmatrix}
    \quad\text{and}\quad 
    h_2(\mathbb{X})=\frac{1}{\sqrt{2}}\begin{pmatrix}
    \phantom{-}0.1 \\ -0.1
    \end{pmatrix}
    \quad\text{with}\quad
    \|h_1(\mathbb{X})\|_2=\|h_2(\mathbb{X})\|_2=\sqrt{\frac{0.02}{2}}=0.1 \ .
\end{align}
For any tolerance $\epsilon \ge 0.1$ the generators $h_1$ and $h_2$ are $\epsilon$-vanishing on $\XX$, coefficient-normalized, and $\langle h_1,h_2\rangle$ is an $\epsilon$-vanishing ideal of $\XX$. The number of order ideal terms is less than the number of prescribed points.

Any $\{1\}$-border prebasis consists of $g_x=C\cdot(x-a)$ and $g_y=D\cdot(y-b)$ with coefficients $C,D>0$. If $g_x$ and $g_y$ are $\epsilon$-vanishing on $\XX$ then the exactly vanishing ideal of the point $(a,b)$ is an $\epsilon$-vanishing ideal of $\XX$.
\item What about approximately vanishing generators that generate the unit ideal? Note $y-a$ and $y-b$ linearly combine to the constant polynomial $b-a$, but these polynomials cannot be part of border basis generators as they use the same border term. 

Instead consider the order ideal $\mathcal{O}=\{1,x\}$ with border $\{x^2,y,xy\}$. As $\XX$ consists of two points, there is the exactly vanishing quadratic polynomial $g_{x^2}=C\cdot(x-0.9)(x-1.1)$ with positive coefficient $C>0$. As $g_{x^2}$ vanishes exactly on $\XX$, its normalization is irrelevant. We construct an $\epsilon$-vanishing border basis that generates the unit ideal. Let $g_y=(y-1)/\sqrt{2}$. It is coefficient-normalized with $\enorm{g_y(\XX)}=0.1$. Consider the slightly offset multiple $g_{xy}=E\cdot \bigl(x(y-1)+0.1\bigr)$. For coefficient normalization, we need $E=1/\sqrt{2.01}$. Its evaluation is $\enorm{g_{xy}(\XX)}\approx 0.13$. For $\epsilon=0.14$ all polynomials are $\epsilon$-vanishing. They generate the unit ideal, since $\sqrt{2.01}g_{xy}-\sqrt{2}x\cdot g_y=0.1\cdot \sqrt{2.01}$ is constant. 
\end{enumerate}
\end{example}

These examples are typical of the general situation. In the exact case, the number of ideal order terms equals the number of prescribed points. In the approximate case, the number of ideal order terms may be smaller than the number of prescribed points. Here, this happens because some prescribed points can be considered as different perturbations of a common zero of the approximately vanishing ideal. 

Alhough we concentrate our attention on border bases, Example~\ref{ex:not-an-order-ideal} below shows that the so-called ABM algorithm need not always produce an order ideal. However, it always produces a set of terms that is connected to the constant term 1 in the following meaning.

\begin{definition}[\cite{elkadi2005symbolic}] A set of terms, $\mathcal{C}\subset \TT$, is \textbf{connected to 1}, if $1\in \mathcal{C}$ and if for each term $1\ne t\in \mathcal{C}$ there are an indeterminate $x_i$ and another term $\tilde{t}\in \mathcal{C}$ such that $t=x_i\cdot \tilde{t}$. For convenience, we refer to the members of a set connected to 1 as \textbf{inner terms}.
\end{definition}

Every order ideal is connected to 1, while the converse is false. The term set $\{1,x,xy\}$ is connected to 1, but not an order ideal. Whenever we want to go beyond the order ideal setting, we still use the symbol $\mathcal{O}$, but call the elements inner terms instead of order ideal terms. All relevant concepts above generalize from order ideals to sets of inner terms, which are connected to 1.

\section{The Gradient-Weighted Semi-Norm of Polynomials}\label{sec:gradient-weighted-semi-norm}

Let $\XX=\{p_1,\dots,p_m\}\subset\RR^n$ be a non-empty, finite set of points. With respect to $\XX$ we define the gradient-weighted semi-norm for polynomials in two stages: first, for terms, then, for polynomials. 

\begin{definition}\label{def:gradient-semi-norm-of-terms}
For $\XX$, the \textbf{gradient-weighted semi-norm} of a term $t\in \TT$ is 
\begin{align}
    \gwnorm{t} =\begin{cases}
        \frac{1}{D(t)}\cdot\enorm{\nabla t(\XX)} & \text{if } t\ne 1\ ,\\
        0 & \text{if }t=1\ .
    \end{cases}
\end{align}
The scalar $D(t)=(\sum_{i=1}^n \deg_i(t)^2)^{1/2}$ denotes the \textbf{Euclidean degree} of $t$. 
\end{definition}

The term-dependent scaling by the Euclidean degree $D(t)$ in the denominator leads to more convenient inequalities later in this paper. In more detail, the gradient-weighted semi-norm is
\begin{align}
    \frac{1}{D(t)}\Bigl(\sum_{j=1}^m \enorm{\nabla t(p_j)}^2\Bigr)^{1/2}
    =\frac{1}{D(t)}\sqrt{\sum_{j=1}^m \sum_{i=1}^n \;\bigl|\frac{\partial t}{\partial x_i}(p_j)\bigr|^2}
    =\frac{1}{D(t)}\left(\sum_{i=1}^n \enorm{\frac{\partial t}{\partial x_i}(\XX)}^2\right)^{1/2}\ .
\end{align}

For example, the gradient-weighted semi-norms of the indeterminates depend only on the number of points in $\XX$, 
\begin{align}
    \gwnorm{x_k}=\frac{1}{D(x_k)}\sqrt{\sum_{p\in\XX} \sum_{i=1}^n \;\bigl|\frac{\partial x_k}{\partial x_i}(p)\bigr|^2}=\frac{1}{1}\sqrt{\sum_{p\in \XX} 1}=\sqrt{|\mathbb{X}|}\ .
\end{align}
From degree two on, the semi-norms depend on the coordinates of the points in $\XX$.

\begin{example}
Let $\mathbb{X}=\{(0,1),(0,3)\}$. Independent of $\XX$ we have $\gwnorm{1}=0$ and $\gwnorm{x}=\gwnorm{y}=2$. For $\ell\ge 2$ we compute
    \begin{align}
        \gwnorm{y^\ell}=\frac{1}{\ell}\enorm{\begin{pmatrix} 0(p_1)\\ \ell y^{\ell-1}(p_1)\\0(p_2) \\ \ell y^{\ell-1}(p_2)\end{pmatrix}}=\frac{1}{\ell}\enorm{\begin{pmatrix}
            0\\ \ell\\ 0\\ \ell 3^{\ell-1}
        \end{pmatrix}}=\sqrt{1+9^{\ell-1}}\ .
    \end{align}
    For $\ell\ge 1$,
    \begin{align}
    \gwnorm{x y^\ell}=\frac{1}{\sqrt{1+\ell^2}}\enorm{\begin{pmatrix} y^\ell(p_1)\\ \ell x y^{\ell-1}(p_1)\\ y^\ell(p_2)\\ \ell x y^{\ell-1}(p_2)\end{pmatrix}}=\frac{1}{\sqrt{1+\ell^2}}\enorm{\begin{pmatrix}
            1\\ 0\\ 3^\ell\\ 0
        \end{pmatrix}}=\frac{\sqrt{1+9^{\ell}}}{\sqrt{1+\ell^2}}
    \end{align}
    and, for $k\ge 2$ and $\ell\ge 1$,
    \begin{align}
        \gwnorm{x^k y^\ell}=\frac{1}{\sqrt{k^2+\ell^2}}\enorm{\begin{pmatrix} kx^{k-1}y^\ell(p_1)\\ \ell x^k y^{\ell-1}(p_1)\\ kx^{k-1}y^\ell(p_2)\\ \ell x^k y^{\ell-1}(p_2)\end{pmatrix}}=\frac{1}{2}\enorm{\begin{pmatrix}
            0\\ 0\\ 0\\ 0
        \end{pmatrix}}=0\ , \quad\gwnorm{x^2}=\frac{1}{2}\enorm{\begin{pmatrix} 2x(p_1)\\ 0(p_1)\\2x(p_2)\\ 0(p_2)\end{pmatrix}}=0
    \end{align}
    The gradient-weighted semi-norm of some terms is zero.
\end{example}

Next, we extend the definition of the gradient-weighted semi-norm to a polynomial. The gradient-weighted semi-norm is a weighted version of the coefficient norm whose weights are the gradient-weighted semi-norms of the terms.

\begin{definition}
The \textbf{gradient-weighted semi-norm} of a polynomial $h = \sum_{t\in \TT} c_t\, t$ with coefficients  $c_t\in\RR$ is defined by 
\begin{align}\label{eq:gradient-weighted-norm}
    \gwnorm{h} = \Bigl(\sum\nolimits_{\,t\in\TT} c_t^2\, \gwnorm{t}^2\Bigr)^{1/2}
    \ .
\end{align}
If $\gwnorm{h}=1$, then $h$ is \textbf{gradient-weighted normalized}.
\end{definition}

\begin{proposition}
The gradient-weighted semi-norm is indeed a semi-norm on $\PP$. For all scalars $\alpha$ and all polynomials $f,g$ we have
\begin{align}
    (i)\quad \gwnorm{f}\ge 0\ ,\quad (ii)\quad \gwnorm{\alpha f}=|\alpha|\cdot\gwnorm{f}\ ,\quad \text{and}\quad
    (iii)\quad \gwnorm{f+g}\le \gwnorm{f}+\gwnorm{g}\ .
\end{align}
The gradient-weighted semi-norm of a polynomial is zero if and only if all terms in its support have zero gradient-weighted semi-norm.
\end{proposition}

\begin{proof}
(i) and (ii) follow immediately from the definition. To prove (iii) let
$f, g \in \PP$ and $f = \sum_{t\in\TT}f_t\, t$ and $g=\sum_{t\in\TT}g_t\, t$ with coefficients $f_t, g_t\in\RR$ of which only finitely many are nonzero. 
We have
\begin{align}
    \gwnorm{f + g}^2
    &= \sum_{t\ne 1}(f_t + g_t)^2\gwnorm{t}^2\\
    &= \sum_{t\ne 1} f_t^2 \gwnorm{t}^2 + \sum_{t\ne 1} g_t^2\gwnorm{t}^2 + 2\sum_{t\ne 1} f_t\gwnorm{t}\cdot g_t \gwnorm{t} \\ 
    &\le \gwnorm{f}^2 + \gwnorm{g}^2 + 2 \sqrt{\sum_{t\ne 1}f_t^2\gwnorm{t}^2}\cdot\sqrt{\sum_{t\ne 1}\gwnorm{t}^2},\qquad\text{(Cauchy-Schwarz)}\\
    &= (\gwnorm{f}+ \gwnorm{g})^2.
\end{align}
The characterization of the non-zero semi-norm value follows immediately from the definition that uses a linear combination of the gradient-weighted semi-norms of the terms with non-zero coefficients.
\end{proof}

Non-zero constant polynomials have zero gradient-weighted semi-norm, but we generally do not consider them as vanishing polynomials. The gradient-weighted semi-norm of a polynomial that includes a linear term is non-zero regardless of the points in $\XX$. If a non-constant polynomial does not contain any linear term, then its gradient-weighted semi-norm may be zero. But this requires in each point of $\mathbb{X}$ some zero coordinate:

\begin{example}\label{example:zero-gcnorm}
Consider the set of points $\XX = \{(1, 0), (0, 1)\}$ in $\RR^2$ and the polynomial
$f = x^2y^2 - 1$ in $\RR[x,y]$. Its gradient-weighted semi-norm is 
\begin{align*}
    \gwnorm{f}&=\sqrt{\frac{1}{2^2+2^2}\gwnorm{x^2y^2}^2}
    =\sqrt{\frac{1}{8}\sum_{p\in\XX} \bigl(|2xy^2(p)|^2+|2x^2y(p)|^2\bigr)}=0
\end{align*}
It is no coincidence that the constant-fee part of $f$, namely $\tilde f=x^2y^2$, vanishes exactly on $\XX$. See Proposition~\ref{prop:zero-gwnorm}.
\end{example}

Thus, depending on $\XX$ there may be non-constant polynomials with zero gradient-weighted semi-norm. However, such polynomials will not occur in our computations to come. The technical foundation for this follows now.

\begin{proposition} %
\label{prop:zero-gwnorm}
Let $\XX=\{p_1,\dots,p_m\}\subset\RR^n$ be a finite, non-empty set.
\begin{enumerate}[(i)]
    \item If a non-constant term has zero gradient-weighted semi-norm, it vanishes exactly on $\XX$,
    \begin{align}
    1\ne t\in \TT\quad\text{with}\quad\gwnorm{t}=0\quad\text{implies}\quad t(\XX)=0\ .
    \end{align}
    \item If a non-constant polynomial $f=\sum_{t\in\TT} c_t\, t$ has zero gradient-weighted semi-norm, then its constant-free part $\tilde f=\sum_{t\ne 1} c_t\, t$ vanishes on $\XX$. 
\end{enumerate}
\end{proposition}

\begin{proof}
\begin{enumerate} [(i)]
\item Let $t=x_1^{e_1}\cdots x_n^{e_n}$ with $e_i>0$ for some $1\le i\le n$. By assumption $\enorm{\nabla t(\XX)}/D(t)=\gwnorm{t}=0$, so $\partial t/\partial x_i(\XX)=0$, i.e., $\partial t/\partial x_i(p_j)=0$ for all $p_j$. Then, $t(p_j)=(e_i^{-1}\cdot x_i\cdot \partial t/\partial x_i)(p_j)=0$ for all $p_j$ and, therefore, $t(\XX)=0$.
\item By assumption $0=\gwnorm{f}=(\sum_{t\ne 1} c_t^2/D(t)^2\cdot \gwnorm{t})^{1/2}$. Thus, each non-constant term $t$ occurring in $f$ has zero gradient-weighted semi-norm. By (i), $t(\XX)=0$ for all these terms, and the result follows.
\end{enumerate}
\end{proof}

Normalization is needed for approximately vanishing polynomials. The use of a semi-norm instead of a positive definite norm possesses the problem that some non-zero polynomials cannot be normalized. However, Theorem~\ref{thm:GWN_as_GEP} below uses Proposition~\ref{prop:zero-gwnorm} to show that all terms and polynomials considered by our computation algorithm are gradient-weighted normalizable.

\section{Approximate Computation of Border Bases with Gradient-Weighted Normalization}
\label{sec:implementing-gradient-normalization}

This section illustrates how gradient-weighted normalization is implemented into existing approximate border basis computation algorithms for $\epsilon$-vanishing ideals. Almost all these computation algorithms rely on solving eigenvalue problems or on computing singular value decompositions (SVD). We implement gradient-weighted normalization by replacing the eigenvalue problem or the SVD with a generalized eigenvalue problem. Section \ref{sec:GWN-as-GEP} discusses this generalized eigenvalue problem.

As a concrete example, Section \ref{sec:AMB-with-gwn} implements gradient-weighted normalization into the approximate Buchberger M\"oller (ABM) algorithm that is described by \cite{limbeck2013computation} and which aims at computing an approximate border basis. We name the adapted algorithm ABM+GWN (gradient weighted normalization). The algorithm determines step-by-step (i) a set of non-$\epsilon$-vanishing terms and (ii) a set of $\epsilon$-vanishing polynomials on $\XX$. We prove its efficacy. \cite{limbeck2013computation} mentions that the ABM algorithm sometimes fails to provide an order ideal but at least outputs a term set connected to 1. Section~\ref{subsec:counterexample} supplies an example of this shortcoming. Alhough we also aim at approximate border bases, the case of a non-order ideal still provides a useful basis of an approximately vanishing ideal.

Section~\ref{sec:gradient-weighted-features} uses the ABM+GWN algorithm to prove the features added by gradient-weighted normalization.

\subsection{The Generalized Eigenvalue Problem of Gradient-Weighted Normalization}\label{sec:GWN-as-GEP}

This subsection analyzes the key computational tool of the ABM+GWN algorithm: solving a generalized eigenvalue problem, whose main purpose is obtaining the solution of the following minimization problem with constraints.

\medskip
\noindent\textbf{Input:} A finite, non-empty set of points $\XX=\{p_1,\dots,p_m\}\subset \RR^n$ and a tolerance $\epsilon\ge 0$ a tolerance. A set of terms $\mathcal{O}=\{t_1,\dots,t_s\}\subset\TT$ so that all gradient-weighted normalized polynomials supported by $\mathcal{O}$ are non-$\epsilon$-vanishing. (Often---but not always---$\mathcal{O}$ is an order ideal.) Moreover, let $b\in\TT\setminus\mathcal{O}$ be a \textbf{trial term}. 

\smallskip
\noindent\textbf{Output:} Either all gradient-weighted normalized polynomials supported by the augmented set $\mathcal{O}\cup\{b\}$ are non-$\epsilon$-vanishing, or there is a gradient-weighted normalized, $\epsilon$-vanishing polynomial $g_{\min}$ with minimal evaluation on $\XX$,
\begin{align}
    g_{\min}=\arg\min\{\,\enorm{g(\XX)}\mid g=v_1 t_1+\dots +v_s t_s+ v b\ \text{with}\  v>0,\ v_1,\dots, v_s\in \RR, \text{and}\ \gwnorm{g}=1\}\ .
\end{align}

All details, in particular, the questions of which polynomials can be gradient-weighted normalized and how to deal with those that cannot, are answered in Theorem~\ref{thm:GWN_as_GEP}.

Now, we translate the constrained minimization problem into an equivalent generalized eigenvalue problem. The evaluation of a polynomial is written as the matrix-vector product
\begin{align}
    g(\XX)=
    \begin{pmatrix}
    g(p_1)\\
    g(p_2)\\
    \vdots\\
    g(p_m)\\
    \end{pmatrix}
    =
    \begin{pmatrix}
    t_1(p_1) & \dots & t_s(p_1) & b(p_1) \\
    t_1(p_2) & \dots & t_s(p_2) & b(p_2)\\
    \vdots && \vdots & \vdots\\
    t_1(p_m) & \dots & t_s(p_m) & b(p_m)\\
    \end{pmatrix}
    \begin{pmatrix}
        v_0 \\ \vdots \\ v_s\\ v
    \end{pmatrix}
    =\bigl(\mathcal{O}(\XX) \ b(\XX) \bigr)\mathbf{v}
    = M \mathbf{v}
\end{align}
and its squared Euclidean norm is the quadratic form
\begin{align}
    0\le\enorm{g(\XX)}^2=\enorm{M\mathbf{v}}^2=\mathbf{v}^\intercal M^\intercal M \mathbf{v}\ .
\end{align}

Similarly, the squared constraint $\gwnorm{g}^2=1$ turns into the quadratic form equation
\begin{align}
    \mathbf{v}^\intercal D^\intercal D \mathbf{v}=1 \quad\text{with the diagonal matrix}\quad D=\begin{pmatrix}
        \gwnorm{t_1} & 0 & \dots & 0 & 0\\
        0 & \gwnorm{t_2} & & 0& 0\\
        \vdots & & \ddots &  &\vdots\\
        0 & 0&&\gwnorm{t_s} & 0\\
        0 & 0& \dots & 0 & \gwnorm{b}
    \end{pmatrix}\ .
\end{align}

We solve the equivalent minimization problem of the squared quantities,
\begin{align}
    \min \bigl\{\,\|g(\XX)\|_2^2\mid g\in\PP, \text{supp}(g)\subseteq \mathcal{O}\cup \{b\}\bigr\} \quad\text{constrained by}\quad\gwnorm{g}^2=1\ ,
\end{align}
using Lagrange multipliers. Thus, we have to find the critical points of 
\begin{align}
    \mathcal{L}(\mathbf{v}, \lambda)=\|g(\XX)\|_2^2-\lambda\cdot (\gwnorm{g}^2-1)
    =\mathbf{v}^\intercal
    M^\intercal M
    \mathbf{v}
    -\lambda
    \mathbf{v}^\intercal
    D^\intercal D
    \mathbf{v}
    +\lambda\ .
\end{align}
The partial derivatives with respect to the components of $\mathbf{v}$ lead to the generalized eigenvalue problem
\begin{align}\label{eq:generalizedEigenproblem}
    M^\intercal M\mathbf{v}=\lambda D^2\mathbf{v}
\end{align}
while the partial derivative with respect to the Lagrange multiplier~$\lambda$ produces the gradient-weighted normalization condition 
\begin{align}\label{eq:constraint}
1=\mathbf{v}^\intercal D^2\mathbf{v}=\gwnorm{v_1t_1+\dots+v_st_s+v b}^2\ .
\end{align}
We multiply equation \eqref{eq:generalizedEigenproblem} by $\mathbf{v}^\intercal$ from the left and use equation \eqref{eq:constraint} to obtain
\begin{align}
    \enorm{g(\XX)}^2=\mathbf{v}^\intercal M^\intercal M\mathbf{v}=\lambda\ .
\end{align}
Therefore the minimization problem is solved by the minimal eigenvalue $\lambda_{\min}$ and a gradient-weighted normalized eigenvector $\mathbf{v}_{\min}$. 

This is the big picture. Now we consider the technical details.

\begin{theorem} \label{thm:GWN_as_GEP}
Let $\XX$ be a non-empty, finite set of points, $\mathcal{O}=\{t_1,\dots,t_s\}$ a set of terms, and $b\notin \mathcal{O}$ another term. These quantities define the generalized eigenvalue problem
\begin{align}
    M^\intercal M\mathbf{v}=\lambda D^2\mathbf{v}\ ,
\end{align}
where $M=\bigl(\mathcal{O}(\XX)\ b(\XX)\bigr)$ is the evaluation matrix, $D$ is the positive semi-definite diagonal matrix with diagonal entries $\bigl(\gwnorm{t_1},\dots,\gwnorm{t_s},\gwnorm{b} \bigr)$, and
\begin{align}
    g=v_1 t_1+\dots+v_s t_s+v b \quad\text{corresponds to}\quad \mathbf{v}=(v_1,\dots, v_s, v)^\intercal\ .
\end{align}
\begin{enumerate}[(i)]
\item \textbf{(Exactly Vanishing.)} The generalized eigenvalue problem possesses the generalized eigenvalue zero, i.e., the matrix $M^\intercal M$ is positive semi-definite, but not positive definite, if and only if there exists a non-zero, exactly vanishing polynomial on $\XX$ that is supported by $\mathcal{O}\cup \{b\}$.

\item \textbf{(Minimal Non-Exactly Vanishing.)} If there is no non-zero exactly vanishing polynomial supported by $\mathcal{O}\cup\{b\}$ then all non-constant polynomials supported by $\mathcal{O}\cup \{b\}$ can be gradient-weighted normalized. In that case the generalized eigenvalue problem corresponds to the minimization problem
\begin{align}
    \min_g \bigl\{\,\|g(\XX)\|_2\mid \text{supp}(g)\subseteq \mathcal{O}\cup \{b\}\bigr\} \quad\text{constrained by}\quad\gwnorm{g}=1\ .
\end{align}
Its minimum $\sqrt{\lambda_{\min}}>0$ corresponds to a minimal solution $(\lambda_{\min}, \mathbf{v}_{\min})$ of the generalized eigenvalue problem,
\begin{align}
    \sqrt{\lambda_{\min}}=\enorm{v_{{\min},1} t_1+\dots+v_{{\min},s} t_s+v_{{\min}} b}
    \quad\text{with gradient-weighted normalization}\quad \mathbf{v}_{\min}^\intercal D^2\mathbf{v}_{\min}=1\ .
\end{align}

\item \textbf{($\mathbf{\epsilon}$-Vanishing.)} Further, let $\epsilon\ge 0$ be a prescribed tolerance and assume that all gradient-weighted normalized polynomials supported by $\mathcal{O}$ are not $\epsilon$-vanishing. If $\sqrt{\lambda_{\min}}\le\epsilon$, then the gradient-weighted normalized $\mathbf{v}_{\min}$ is uniquely determined up to sign, and the coefficient $v$ is non-zero. We choose the sign that makes the coefficient $v$ of the term $b$ positive. The polynomial $g_b=v_{{\min},1} t_1+\dots+v_{{\min},s} t_s+v_{{\min}} b$ is gradient-weighted normalized and $\epsilon$-vanishing on $\XX$.
\end{enumerate}
\end{theorem}

\begin{proof}
The diagonal matrix $D$ has non-negative diagonal elements, so it is positive semi-definite.
\begin{enumerate}[(i)]
    \item The symmetric matrix $M^\intercal M$ is positive semi-definite because of $0\le \enorm{g(\XX)}^2=\mathbf{v}^\intercal M^\intercal M\mathbf{v}$. This is the left-hand side of the generalized eigenvalue problem
    \begin{align}\label{eq:evalution-by-eigenvalue}
        \mathbf{v}^\intercal M^\intercal M\mathbf{v} =\lambda\cdot \mathbf{v}^\intercal D^2\mathbf{v}\ .
    \end{align} 
    If there is a generalized eigenvalue $\lambda=0$ then there is a generalized eigenvector $\mathbf{v}\ne \mathbf{0}$ that satisfies $0=\mathbf{v}^\intercal M^\intercal M\mathbf{v}=\enorm{g(\XX)}$.
    Conversely, if $g\ne 0$ vanishes exactly on $\XX$, i.e., $\enorm{g(\XX)}=0$, the coefficient vector $\mathbf{v}$ of $g$ is non-zero, and the generalized eigenvalue equation is satisfied with $\lambda=0$.
    \item Assume that there is a non-constant polynomial $g$ supported by $\mathcal{O}\cup\{b\}$ that cannot be gradient-weighted normalized, i.e., $\gwnorm{g}=0$. By Proposition~\ref{prop:zero-gwnorm}, the constant-free part of $g$ is non-zero and vanishes exactly on $\XX$, which contradicts the assumption of the case (ii).
    
    The correspondence between the minimization problem and the generalized eigenvalue problem has been derived in the text preceding this theorem. Since the generalized eigenvalue corresponds to the squared quantities of the minimization problem, the evaluation of the minimal polynomial is the square root of $\lambda_{\min}$.
    
    \item  We only have to show that, under the additional assumptions, the generalized eigenspace of the minimal generalized eigenvalue is one-dimensional. Note that a gradient-weighted normalized polynomial must be non-constant.

Let $\sqrt{\lambda_{\min}}\le \epsilon$ and $\mathbf{u}$ and $\mathbf{v}$ be eigenvectors of the minimal eigenvalue $\lambda_{\min}$. The coefficients $u$ and $v$ of the trial term $b$ must be non-zero. Otherwise, the corresponding polynomials would be supported by $\mathcal{O}$ and contradict the non-$\epsilon$-vanishing case assumption.

The linear combination $\mathbf{w}=v\cdot  \mathbf{u}-u\cdot\mathbf{v}$ eliminates the coefficient of the trial term $b$, and hence the corresponding polynomial is supported by $\mathcal{O}$. Its evaluation $\sqrt{\lambda_{\min}}$ exceeds $\epsilon$ by assumption. On the other hand, the generalized eigenvalue description gives $M^\intercal M \mathbf{w}=\lambda_m D^2\mathbf{w}$. To avoid a contradiction, we must have $\mathbf{w}=\mathbf{0}$. So the eigenvectors are linearly dependent, and the eigenspace of $\lambda_{\min}$ is one-dimensional. 
\end{enumerate}
\end{proof}

\begin{remark} [\textbf{Infinite Generalized Eigenvalue}]
The matrix $D$ is positive semi-definite but can fail to be positive definite. In particular, this is the case when $D$ contains the diagonal entry $\gwnorm{1}=0$. In this non-definite case we may have a non-zero vector $\mathbf{v}$ with
\begin{align}
    M^\intercal M\mathbf{v}\ne 0 \quad\text{and}\quad \lambda \underbrace{D^2\mathbf{v}}_{=0}=0\ .
\end{align}
We follow the general convention and say that $\mathbf{v}$ belongs to the generalized eigenvalue $\lambda=\infty$. The generalized eigenvalue $\lambda=\infty$ occurs when the degree of the characteristic polynomial $\det(M^\intercal M)$ is smaller than the size of $M^\intercal M$. (A proper treatment uses the homogenization $\lambda_0 M^\intercal M\mathbf{v}=\lambda_1 D^2\mathbf{v}$. Then the value infinity belongs to the homogeneous coordinate $[0:\lambda_1]$ of the point at infinity on the projective line.) 
\end{remark}

\begin{example} \label{example:eigenvalue-problem}
(\textbf{Exactly Vanishing.}) Let $\XX=\{(2,1)\}$ consist of a single point, $\mathcal{O}=\{1\}$ the set of terms, and  $b=x$ the trial term. Then
\begin{align}
    M=\bigl(1(\XX)\; x(\XX)\bigr)=\bigl( 1\; 2\bigr)\quad\text{and}\quad
    M^{\intercal}M =\begin{pmatrix}
        1& 2\\ 2 & 4
    \end{pmatrix}\quad.
\end{align}
Also
\begin{align}
    D=\begin{pmatrix}
        \gwnorm{1} & 0 \\
        0 & \gwnorm{x}
    \end{pmatrix}
    \quad\text{with}\quad
    \gwnorm{1}=0
    \quad\text{and}\quad 
    \gwnorm{x}=\frac{\enorm{\nabla x(\XX)}}{D(x)}=\frac{\enorm{1\choose 0}}{1}=1\ .
\end{align}
The generalized eigenvalue problem is
\begin{align}
    \begin{pmatrix}
        1 & 2 \\ 2 & 4
    \end{pmatrix}
    \begin{pmatrix}
        v_1 \\ v
    \end{pmatrix}
    =\lambda
    \begin{pmatrix}
        0 & 0 \\  0 & 1
    \end{pmatrix}
    \begin{pmatrix}
        v_1 \\ v
    \end{pmatrix}
\end{align}
with normalization condition 
    $\mathbf{v}D^{\intercal}D \mathbf{v}=1$, i.e. $v^2=1$.
The characteristic equation has degree one,
\begin{align}
    0=\det\begin{pmatrix}
        1 & 2 \\ 2 & 4-\lambda
    \end{pmatrix}=(4-\lambda)-4=-\lambda\ .
\end{align}
The generalized eigenvalues are $\lambda=\infty$ with generalized eigenvector $(1,0)^\intercal$ and $\lambda=0$ with generalized eigenvector direction $\mathbf{v}=\gamma\cdot (-2\; 1)^\intercal$, $\gamma\ne 0$. The generalized eigenvalue zero signals that the generalized eigenvector corresponds to an exactly vanishing polynomial (which does not need normalization). Anyway, gradient-weighted normalization is possible. Then the eigenvector is determined up to sign, $\mathbf{v}=\pm (-2\; 1)^\intercal$. We choose the sign that makes the coefficient of the trial term $b$ positive. The polynomial $g= -2+x$ is gradient-weighted normalized, has a positive coefficient of the trial term $x$, and vanishes exactly on $\XX$.
\end{example}

The following example uses the Matlab command \texttt{[V,$\Lambda$]=Eig(M'*M,D'*D)} for solving the generalized eigenvalue problem.

\begin{example} (\textbf{Minimal Non-Exactly Vanishing.})
 Let $\XX=\{(2,1), (4.1,2), (6,3)\}$ be three almost collinear points, $\mathcal{O}=\{1\}$, and  trial term $b=x$. Then
\begin{align}
    M=\bigl(1(\XX)\; x(\XX)\bigr)=\begin{pmatrix}
        1 & 2 \\ 1 & 4.1 \\ 1 & 6
    \end{pmatrix}\quad\text{and}\quad
    M^{\intercal}M =\begin{pmatrix}
        3 & 12.1 \\ 12.1 & 56.81
    \end{pmatrix}
\end{align}
and
\begin{align}
    D=\begin{pmatrix}
        \gwnorm{1} & 0 \\
        0 & \gwnorm{x}  
    \end{pmatrix}
    \quad\text{with}\quad
    \gwnorm{x}=\frac{\enorm{\nabla x(\XX)}}{D(x)}=\frac{\enorm{(1\,0\,1\,0\,1\,0)^{\intercal}}}{1}=\sqrt{3}
    \quad.
\end{align}
The generalized eigenvalue problem and its normalization constraints are
\begin{align}
    \begin{pmatrix}
        3 & 12.1 \\ 12.1 & 56.81
    \end{pmatrix}
    \begin{pmatrix}
        v_1 \\  v
    \end{pmatrix}
    =\lambda
    \begin{pmatrix}
        0 & 0 \\  0 & 3
    \end{pmatrix}
    \begin{pmatrix}
        v_1 \\ v
    \end{pmatrix}
    \quad\text{such that}\quad  \begin{pmatrix}
        v_1 &  v
    \end{pmatrix}
    \begin{pmatrix}
        0 & 0 \\  0 & 3
    \end{pmatrix}
    \begin{pmatrix}
        v_1 \\  v
    \end{pmatrix}=1\quad\text{, i.e.,}\quad 3 v^2=1
    \quad\text{and}\quad v>0\ .
\end{align}
The characteristic equation has degree one,
\begin{align}
    0=\det\begin{pmatrix}
        3 & 12.1 \\ 12.1 & 56.81-3\lambda 
    \end{pmatrix}=(56.81-3\lambda)\cdot 3-146.41=24.02-9\lambda\ .
\end{align}
There is only one finite generalized eigenvalue: $\lambda_{\min}=24.02/9\approx 1.634$.   Matlab computes the eigenvector $({v}_1,v)^{\intercal}\approx (-1, 0.2479)^{\intercal}$, which still has to be gradient-weighted normalized. We compute
\begin{align}
\gwnorm{-1+0.2479x}=
\left[\begin{pmatrix}
    -1 & 0.2479
\end{pmatrix}
\begin{pmatrix}
    0 & 0 \\ 0 & 3
\end{pmatrix}
\begin{pmatrix}
    -1 \\ 0.2479
\end{pmatrix}\right]^{1/2}
\approx 0.4294\quad\text{and}\quad
\frac{1}{0.4294}\cdot \begin{pmatrix}
    -1 \\ 0.2479
\end{pmatrix}
\approx 
\begin{pmatrix}
    -2.3288 \\
    0.5773
\end{pmatrix}
\ .
\end{align}
The polynomial $g_x\approx 0.5773x-2.3288$ has Euclidean evaluation norm 
\begin{align}
    \enorm{g_x(\XX)}\approx\enorm{\begin{pmatrix}
        0.5773\cdot 2-2.3288\\
        0.5773\cdot 4.1-2.3288\\
        0.5773\cdot 6-2.3288
    \end{pmatrix}}
    \approx \enorm{\begin{pmatrix}
        -1.1742\\
        0.0381\\
        1.1350
    \end{pmatrix}}
    \approx 1.634\approx \sqrt{\lambda_{\min}}\ .
    \end{align}
\end{example}

\subsection{The ABM algorithm with gradient-weighted normalization: the ABM+GWN algorithm}\label{sec:AMB-with-gwn}

\begin{algorithm}[t]
\caption{\\ The ABM+GWN algorithm (Approximate Buchberger-M\"oller with gradient-weighted normalization)}\label{alg:ABMGN}
\DontPrintSemicolon
  \KwInput{Point set $\mathbb{X}$, tolerance $\epsilon$, and a degree-compatible term ordering $\sigma$}
  \KwOutput{(i) List of terms $\mathcal{O}$ connected to 1. All gradient-weighted normalized polynomials supported by $\mathcal{O}$ are not $\epsilon$-vanishing. (ii) For each term $b\in\partial\mathcal{O}$ a polynomial $g_b$ with support in $\{b\}\cup\mathcal{O}$. The border term $b$ has a positive coefficient, and $g_b$ is gradient-weighted normalized and $\epsilon$-vanishing on $\XX$. }
   $\mathcal{O} \leftarrow [1]$\\
   $s\leftarrow 1$ \quad \% length of $\mathcal{O}$\\
   $G \leftarrow []$\\
   $r\leftarrow 0$ \quad \% length of $G$\\
   $d\leftarrow 1$ \quad \% current degree of terms\\
   \% list of trial terms in degree $d$:\\
   $T_d \leftarrow [b_1,\dots, b_{n_d}]$ 
  with $b_i\in \partial\mathcal{O}$, $\deg b_i=d$, and $b_1<_\sigma\dots<_\sigma b_{n_d}$\\
    \While {$T_d\ne [\;]$}{
  \For{$b$ in $T_d$}{
  Compute $\lambda_{\min}$ and  $\mathbf{v}_{\min}=(v_0,\ldots,v_{s})^{\intercal}$
  of the generalized eigenvalue problem
  $M^{\intercal}M\mathbf{v}=\lambda D^2\mathbf{v}$\\
  for $\mathcal{O}=[t_1,t_2,\cdots, t_s]$ and $b$ with $M=\bigl(\mathcal{O}(\XX)\;b(\XX)\bigr)$
  and $D=\text{diag}\bigl(\gwnorm{o_1},\dots,\gwnorm{o_s},\gwnorm{b}\bigr)$.\\
  Normalize $\textbf{v}$ so that $v>0$ and $\mathbf{v}^{\intercal}D^{\intercal}D\mathbf{v}=1$.
  \\
      \If{$\sqrt{\lambda_{\min}} \le \epsilon$}{
        Append the polynomial $g_{r+1}= v_1t_1 + \cdots +v_{s}t_{s}+v t$ to the list $G$\\
        $r\leftarrow r+1$\\
      }
      \Else{
        Append the term $t_{s+1}\leftarrow b$ to the list $\mathcal{O}$\\
        $s\leftarrow s+1$\\
      }
      }
      $d\leftarrow d+1$\\
      $T_d\leftarrow [b_1,\dots, b_{n_d}]$
  with $b_i\in \partial\mathcal{O}$, $\deg b_i=d$, and $b_1<_\sigma\dots<_\sigma b_{n_d}$\\
 }
 \Return{$\mathcal{O}$ and $G$}
\end{algorithm}

This subsection shows how gradient-weighted normalization is implemented in an already existing algorithm for computing approximate border bases. The key innovation is the generalized eigenvalue problem as presented in Section~\ref{sec:GWN-as-GEP}. As an example of an existing algorithm, we choose the Approximate Buchberger--M\"oller (ABM) algorithm as described by \cite{limbeck2013computation}. Refer to Appendix~\ref{app:method-comparison} for this choice, and  Appendix~\ref{sec:avi+gwn} for an adaptation with another algorithm.

Algorithm~\ref{alg:ABMGN} shows our adaptation of the ABM algorithm, which implements the generalized eigenvalue problem based on gradient-weighted normalization. We use the abbreviation ABM+GWN. In the following, to fix the order of the points in $\XX$, the terms, and the $\epsilon$-vanishing polynomials, we consider them as items in lists $[\dots]$ instead of as elements in sets $\{\dots\}$.

\medskip
\noindent\textbf{Input:} A non-empty, finite list of points $\mathbb{X}$, a tolerance $\epsilon\ge 0$, and a degree-compatible term ordering $\sigma$.

\smallskip
\noindent\textbf{Output:} \begin{enumerate} [(i)] \item A list of terms $\mathcal{O}$ that is connected to 1. All gradient-weighted normalized polynomials supported by $\mathcal{O}$ are not $\epsilon$-vanishing. 
\item For each term $b\in\partial\mathcal{O}$ a polynomial $g_b$ with support in $\{b\}\cup\mathcal{O}$. The border term $b$ has a positive coefficient, and $g_b$ is gradient-weighted normalized and $\epsilon$-vanishing on $\XX$. 
\end{enumerate}
\medskip

At degree $d=0$, the list of terms $\mathcal{O} = [1]$ of length $s=1$ and the list of the border prebasis polynomials $G=[\;]$ of length $r=0$ are initialized. Since $\XX$ is non-empty, the constant term 1 is generally considered as non-$\epsilon$-vanishing and placed in the term set. 

At degree $d=1$, the degree-$d$ terms belonging to $\partial\mathcal{O}$ are initialized in the list of trial terms $T_d = [b_1,\dots,b_{n_d}]$ in increasing order $ b_1<_\sigma \dots <_\sigma b_{n_d}$. Specifically, $T_1$ is the list of indeterminates $x_1$, \dots, $x_n$ in $\sigma$-increasing order and non-empty.

If the list of trial terms $T_d=[b_1,\dots, b_{n_d}]$ of degree $d$ is non-empty, we will execute the following  for loop over all trial terms $b_i$ in $\sigma$-increasing order.\footnote{\cite{limbeck2013computation} uses $\sigma$-decreasing order. Increasing order has the advantage that a trial term is greater than all terms in $\mathcal{O}$. In this construction, border terms are $\sigma$-leading terms.}

\begin{itemize}  
\item[(F1)] For $M = \mqty(O(\XX) & b(\XX))$ and 
    $D = \mathrm{diag}\qty(\gwnorm{t_1},\ldots, \gwnorm{t_s},\gwnorm{b})$,
    solve the generalized eigenvalue problem
    \begin{align}\label{eq:gep}
        M^{\intercal}M\mathbf{v}_{\min} &= \lambda_{\min}D^2\mathbf{v}_{\min},
    \end{align}
     where $\lambda_{\min}$ is the minimal generalized eigenvalue, possibly zero, and $\mathbf{v}_{\min}=(v_1,\ldots,v_{s},v)^{\intercal}$ is the corresponding generalized eigenvector. 
     
    \item[(F2)] If $\sqrt{\lambda_{\min}} \le \epsilon$, normalize $\textbf{v}$ so that $\textbf{v}^{\intercal}D^{\intercal}D\textbf{v}=1$ with $v>0$. Append the polynomial
    \begin{align}\label{eq:vanishing-polynomial}
        g_{r+1} = v_1t_1 + v_2t_2 + \cdots + v_{s}t_s+v b
    \end{align}
     to the list of $\epsilon$-vanishing border basis polynomials $G$. Increment the length $r$. 
     
     \item[(F3)] Else, if $\sqrt{\lambda_{\min}} > \epsilon$, append $t_{s+1}=b$ to the list of order ideal terms $\mathcal{O}$ and increment its length $s$. 
\end{itemize}

After the for loop (F1)--(F3), increment the degree, $d\leftarrow d+1$. At the new degree $d$, the degree-$d$ terms belonging to current $\partial\mathcal{O}$ initialize the list $T_d = [b_1,\dots,b_{n_d}]$ in increasing order $ b_1<_\sigma \dots <_\sigma b_{n_d}$. If $T_d$ is non-empty, execute the for loop again.

If the degree $d$ is sufficiently large so that the initialized list $T_d$ is empty, then the algorithm will output $(\mathcal{O},G)$ and terminate.

\begin{proposition}\label{prop:invariants-of-ABM+GWN}
The invariants of the for loop of the ABM+GWN algorithm are: 
\begin{enumerate}[(i)]
    \item The list of terms $\mathcal{O}=[t_1,\dots,t_s]$ is connected to 1. 
    
    \item Every non-constant polynomial $f$ supported by ${\mathcal{O}}$ can be gradient-weighted normalized, i.e., $\gwnorm{f}>0$, and then does not $\epsilon$-vanish on $\mathbb{X}$,
    \begin{align}
        \frac{1}{\gwnorm{f}}\cdot\enorm{f(\XX)}>\epsilon\quad ,\ \text{in particular,}\quad
        \frac{1}{\gwnorm{t}}\cdot\enorm{t(\XX)}>\epsilon\quad \text{for all terms \ }1\ne t\in\mathcal{O}\ .
    \end{align}
    
    \item For the polynomials $g_1,\dots, g_r$ in $G$ there are pairwise different terms $b_1,\dots, b_r\in\partial\mathcal{O}$  such that $g_i$ is supported by $\mathcal{O}\cup\{b_i\}$ and the coefficient of $b_i$ is positive. Each $b_i\in\partial{O}$ can be gradient-weighted normalized, 
    \begin{align}
        \gwnorm{b_i}>0\ .
    \end{align}
   Moreover each polynomial $g_i$ can be gradient-weighted normalized and then is $\epsilon$-vanishing on $\XX$,
    \begin{align}
        \frac{1}{\gwnorm{g_i}}\cdot\enorm{g_i(\XX)}\le\epsilon\ .
    \end{align}
\end{enumerate}
\end{proposition}
\begin{proof}
Before the first run of the for loop we have $\mathcal{O}=[1]$ and $G=[]$. All three invariants are trivially satisfied. In each run of the for loop, the generalized eigenvalue problem is solved for the current order ideal $\mathcal{O}=[t_1,\dots,t_s]$ and the trial term $b$. 
\begin{enumerate}[(i)]
    \item If a term $t_{s+1}=b$ is appended to $\mathcal{O}$, it started as a trial term $b\in\partial\mathcal{O}$, which is connected to $\mathcal{O}$. The list $\mathcal{O}$ stays connected to 1.

\item If a term $t_{s+1}=b$ is appended to $\mathcal{O}$, then the generalized eigenvalue problem produced $\sqrt{\lambda_{\min}}>\epsilon$. By Theorem~\ref{thm:GWN_as_GEP}, all non-constant polynomials supported by $\mathcal{O}\cup\{b\}$ can be gradient-weighted normalized. Their evaluations are greater than $\sqrt{\lambda_{\min}}$ and thus greater than $\epsilon$.

\item If a polynomial $g_{r+1}$ is appended to $G$ then it is supported by a trial term $b=b_{r+1}\in\partial\mathcal{O}$ that has not been considered before. Since $b$ is a border term at the time of its consideration, there is an indeterminate $x_i$ and a term $t_j\in \mathcal{O}$ such that $b=x_i\cdot t_j$. Hence $\partial b/\partial x_i= \degreek{i}{b}\cdot t_j$ with $\degreek{i}{b}>0$ and
\begin{align}
    \gwnorm{b}=\frac{\enorm{\nabla b(\XX)}}{{D(b)}}\ge \frac{1}{D(b)}\cdot\enorm{\frac{\partial b(\XX)}{\partial x_i}}=\frac{\degreek{i}{b}}{D(b)}\cdot \enorm{t_j(\XX)}>0\ ,
\end{align}
since $t_j\in\mathcal{O}$.
The relationship $b=x_i\cdot t_j$ with $t_j\in\mathcal{O}$ does not change when $\mathcal{O}$ gets extended, so $b$ will remain a border term of all future $\mathcal{O}$. By construction, all trial terms are pairwise different, so $b_1,\dots, b_{r+1}$ are pairwise different. Since $b$ occurs in $g_i$ with coefficient $c_b>0$, we have
\begin{align}
    \gwnorm{g_i}= \left(\sum_{t\in\supp{g}} c_t^2\, \gwnorm{t}^2\right)^{1/2}\ge c_b\, \gwnorm{b}>0\ .
\end{align}The $\epsilon$-vanishing property of $g_i$ follows from its construction.
\end{enumerate}
\end{proof}

\begin{remark}\label{rem:diff-from-original-ABM}
The ABM+GWN algorithm differs from the ABM algorithm only in one computational detail: the generalized eigenvalue problem $(M^{\intercal}M, D^2)$ in ABM+GWN replaces the SVD of $M$ in ABM. That SVD is equivalent to the eigenvalue problem of $M^{\intercal}M$. Therefore, the change from ABM to ABM+GWN replaces the ordinary eigenvalue problem
\begin{align}
    M^\intercal M\mathbf{v}=\lambda I\mathbf{v} \qquad\text{(identity matrix $I$ inserted for emphasis)}
\end{align}
by the generalized eigenvalue problem
\begin{align}
    M^\intercal M\mathbf{v}=\lambda D^2\mathbf{v}\ .
\end{align}
The matrices $D^2$ and $I$ on the right-hand sides encode the normalization. Obviously, other normalizations can be designed after this pattern.
\end{remark}

The following theorem states the fundamental properties of the ABM+GWN algorithm. There are noteworthy differences to Theorem~4.3.1 in~\citep{limbeck2013computation} about the ABM algorithm. However, the improvements in the formulations of these fundamental properties are not due to the implementation of gradient-weighted normalization. We introduced a new way to extend the key definitions from the exact to the approximate case, whose suitability is shown by simpler statements and proofs. In particular, the ABM algorithm uses the definition of an $\delta$-approximate border basis, which introduces a second parameter, while our definition of an $\epsilon$-approximate border basis for $\XX$ uses the tolerance prescribed by the problem. 
 
The significant advantages of gradient-weighted normalization, in particular, its data-driven approach, are presented in Section~\ref{sec:gradient-weighted-features}.

\begin{theorem}\label{thm:ABM+GWN}
Let $\XX\subset\RR^n$ be a finite, non-empty set of points, $\epsilon\ge 0$ a tolerance, and $\sigma$ a degree-compatible term ordering. The ABM+GWN algorithm~(Algorithm~\ref{alg:ABMGN}) computes in finitely many steps a list of terms $\mathcal{O}\subset\TT$ and a list of polynomials $G=[g_1,\dots,g_r]$ with the following properties:
\begin{enumerate}[(i)]
\item The list of terms $\mathcal{O}$ is connected to 1.

\item Any non-constant polynomial $f$ supported by $\mathcal{O}$ can be gradient-weighted normalized, i.e., $\gwnorm{f}>0$, and then does not $\epsilon$-vanish on $\XX$,
\begin{align}
    \frac{1}{\gwnorm{f}}\cdot \enorm{f(\XX)}>\epsilon\quad ,\ \text{in particular,}\quad
        \frac{1}{\gwnorm{t}}\cdot\enorm{t(\XX)}>\epsilon\quad \text{for all terms \ }1\ne t\in\mathcal{O}\ .
\end{align}

\item There is a one-to-one correspondence between the border terms $\partial\mathcal{O}=[b_1,\dots,b_r]$ and the polynomials of $G=[g_1,\dots,g_r]$. Each polynomial $g_i$ is supported by $\mathcal{O}\cup \{b_i\}$ and $b_i$ has a positive coefficient. Each border term $b_i$ can be gradient-weighted normalized.
Moreover, each polynomial $g_i$ is gradient-weighted normalized and $\epsilon$-vanishing on $\XX$. Therefore, $\mathcal{I}=\langle g_1,\dots, g_r\rangle$ is an $\epsilon$-vanishing ideal of $\XX$ with respect to gradient-weighted normalization.

\item If $\mathcal{O}$ is an order ideal, then $G$ is an $\epsilon$-approximate, gradient-weighted normalized $\mathcal{O}$-border basis for $\XX$. Moreover, if each $g_i$ vanishes exactly on $\XX$, then $G$ is an $\mathcal{O}$-border basis.
\end{enumerate}
\end{theorem}

\begin{proof} We first show the finiteness of the ABM+GWN algorithm. Suppose we are at degree $d$ with $T_d = [b_1,\ldots, b_{n_d}]$. The algorithm processes each term in $T_d$ one by one. If at (F2) all the terms in $T_d$ lead to $\epsilon$-vanishing border basis polynomials and are appended to $G$, then the algorithm terminates because we have an empty list $T_{d+1} = []$ in the next degree. Otherwise, the order ideal $\mathcal{O}$ is extended at least by one term. Once we reach $\abs{\mathcal{O}} = \abs{\mathbb{X}}$, we always have $\lambda_{\min} = 0$ at Eq.~\eqref{eq:gep}. This is because $O(\mathbb{X})$ is now a square and full-rank matrix of size $\abs{\mathbb{X}}$, and thus, $M^{\intercal}M$ with $M = \qty(\mathcal{O}(\mathbb{X})\ b(\mathbb{X}))$ leads to a degenerated square matrix of size $\abs{\mathbb{X}}+1$. Thus, $\sqrt{\lambda_{\min}} \le \epsilon$ at (F2) always holds, and eventually, there will be no trial terms at the next degree, which terminates the algorithm.
The arguments (i), (ii), and (iii) are the results of the invariants of Proposition~\ref{prop:invariants-of-ABM+GWN} at the termination of the algorithm. Property (iv) follows from (ii) and (iii) by definition of an $\epsilon$-approximate border basis for $\XX$ and the definition of an exact border basis for $\XX$.

\end{proof}

\begin{corollary}
At each stage of the ABM+GWN algorithm, the only zero diagonal entry in 
\begin{align}
    D=\bigl(\gwnorm{1},\gwnorm{t_2},\dots,\gwnorm{t_s},\gwnorm{b}\bigr)
\end{align} 
is $\gwnorm{1}$. All other diagonal entries are positive.
\end{corollary}

\begin{example} [\textbf{Middle Tolerance}]\label{ex:ABM+GWN-middle-tolerance}
We consider the point set $\XX = \{p_1, p_2, p_3\}$ with
        $p_1 = (2,2)$, 
        $p_2 = (-4,1)$, and
        $p_3 = (2, 0)$.
We run the ABM+GWN algorithm with tolerance $\epsilon=1$ and the graded reverse lexicographic term ordering $\sigma$, i.e., $y<x$. 
   
Again, we use the Matlab command \texttt{[V,$\Lambda$]=Eig(M'*M,D'*D)} for solving the generalized eigenvalue problems. According to Matlab's output standard, the results are reported with four significant digits. 

\begin{enumerate}
       \item The algorithm starts with $\mathcal{O}= [1]$, $T_1=[y,x]$, $G=[\,]$, and trial term $b= y$. Thus we have 
    \begin{align}
        M_y = \mqty(1 & 2 \\ 
                    1 & 1   \\ 
                    1 & 0 
                    ) \ ,\quad M_y^{\intercal}M_y = \mqty(3 & 3 \\ 3 & 5)\ ,\text{ and }\quad D_y^2=\begin{pmatrix}
        \gwnorm{1}^2 & 0 \\
        0 & \gwnorm{y}^2
    \end{pmatrix}
    =\mqty(0 & 0 \\ 0 & 3)\ .
    \end{align}
    The generalized eigenvalues of $M_y^{\intercal}M_y \mathbf{v}=\lambda D_y^2\mathbf{v}$ are $2/3$ with generalized eigenvector $(-1,1)^\intercal$ and $\infty$. Since $\lambda_{\min} < \epsilon^2$, the generalized eigenvector for the minimal generalized eigenvalue produces the gradient-weighted normalized basis polynomial $g_y=(y-1)/\sqrt{3}\approx 0.5774y-0.5774$. The normalized border term has evaluation $\enorm{y(\XX)}/\gwnorm{y}=\sqrt{5}/\sqrt{3}\approx 1.291$. Check: $\enorm{g_y(\XX)}=\enorm{(1/\sqrt{3},0, -1/\sqrt{3})^\intercal}=\sqrt{2/3}$. 
    
    \item The next step uses $\mathcal{O} = [1]$, $T_1=[y,x]$, $G=[g_y]$, and trial term $b= x$. Therefore 
    \begin{align}
        M_x = \mqty(1 & \phantom{-}2 \\ 
                    1 & -4 \\
                    1 & \phantom{-}2 ) \ ,\quad M_x^{\intercal}M_x = \mqty(3 & 0 \\ 0 & 24)\ ,\text{ and }\quad D_x^2=\begin{pmatrix}
        \gwnorm{1}^2 & 0 \\
        0 & \gwnorm{x}^2
    \end{pmatrix}
    =\mqty(0 & 0 \\ 0 & 3)\ .
    \end{align}
    Since both generalized eigenvalues, 8 and $\infty$, are greater than $\epsilon^2$, the list of terms is augmented to $\mathcal{O} = [1, x]$. We have $\enorm{x(\XX)}/\gwnorm{x}=\sqrt{24}/\sqrt{3}=\sqrt{8}$. This finishes the trial terms of degree one.
    
    \item Now $\mathcal{O} = [1, x]$, $T_2=[xy, x^2]$, $G=[g_y]$, and trial term $b=xy$. We have
    \begin{align}
        M_{xy} = \mqty(1 & \phantom{-}2 & \phantom{-}4\\ 
                       1 & -4 & -4\\
                       1 & \phantom{-}2 & \phantom{-}0
                       ) \ ,\quad M_{xy}^{\intercal}M_{xy} 
              = \mqty(3 & 0 & 0  \\ 
                      0 & 24 & 24 \\
                      0 & 24 & 32
                      )\ ,\text{ and }\quad
                      D_{xy}^2=
    \mqty(0 & 0 & 0\\ 0 & 3 & 0 \\ 0 & 0 & 14.5)
    \end{align}
    with generalized eigenvalues 0.4525, 9.7544, and $\infty$. A generalized eigenvector for the minimal eigenvalue $(0, -1.000, 0.9434)^\intercal$. Since $\lambda_{\min}< \epsilon^2$, we compute the corresponding gradient-weighted normalized basis polynomial $g_{xy}=0.2366xy-0.2507x$. Check: Up to the roundoff error, the evaluation agrees with $\sqrt{\lambda_{\min}}$,
    \begin{align}
       \enorm{g_{xy}(\XX)}=\enorm{\begin{pmatrix}
           0.2366\cdot 2\cdot 2-0.2507\cdot 2\\
           0.2366\cdot (-4)\cdot 1-0.2507\cdot (-4)\\
           0.2366\cdot 2\cdot 0-0.2507\cdot 2
       \end{pmatrix}} 
       =\enorm{\begin{pmatrix}
           0.4450\\ 0.0564 \\ -0.5014
       \end{pmatrix}}\approx \sqrt{0.4526}\approx 0.6728\ .
    \end{align}
    Since the coefficients of $g_{xy}$ differ, this is not the gradient-weighted normalization of $x\cdot g_y\approx 0.5774xy-0.5774x$. We have $\gwnorm{x\cdot g_y}\approx 2.415$. We compute
    \begin{align}
    \frac{\enorm{(x\cdot g_y)(\XX)}}{\gwnorm{x\cdot g_y}}=
    \frac{1}{2.415}
       \enorm{\begin{pmatrix}
           2\cdot (0.5774\cdot 2-0.5774)\\
           -4\cdot (0.5774\cdot 1-0.5774)\\
           2\cdot (0.5774\cdot 0-0.5774)
       \end{pmatrix}} 
       =\frac{1}{2.415}\enorm{\begin{pmatrix}
           1.1548\\ 0 \\ -1.1548
       \end{pmatrix}}\approx 0.6762>\enorm{g_{xy}(\XX)}\ .
    \end{align}
    The border term $xy$ has gradient-weighted normalized evaluation
    \begin{align}
        \frac{\enorm{xy(\XX)}}{\gwnorm{xy}}=\frac{\sqrt{4^2+4^2+0^2}}{\sqrt{14.5}}\approx 1.486\ .
    \end{align}

    \item Now $\mathcal{O} = [1, x]$, $T_2=\{xy,x^2\}$, $G=[g_y, g_{xy}]$, and trial term $b=x^2$. We have
    \begin{align}
        M_{x^2} = \mqty(1 & \phantom{-}2 & 4\\ 
                       1 & -4 & 16\\
                       1 & \phantom{-}2 & 4
                       ) \ ,\quad M_{x^2}^{\intercal}M_{x^2} 
              = \mqty(3 & 0 & \phantom{-}24\\ 
                      0 & 24 & -48\\
                      24 & -48 & 288)
                      \ ,\text{ and }\quad
            D_{x^2}^2=
    \mqty(0 & 0 & 0\\ 0 & 3 & 0 \\ 0 & 0 & 24)\ .
    \end{align}
    with Matlab-computed generalized eigenvalues 0.000, 12.0000, and $\infty$. The resulting border polynomial is $g_{x^2}=0.2222x^2+0.4444x-1.7778$. It is appended to $G$. Up to roundoff error, this is the gradient-weighted normalization of the exactly vanishing polynomial $h_{x^2}=(x-2)(x+4)=x^2+2x-8$. The border term $x^2$ has gradient-weighted normalized evaluation
    \begin{align}
        \frac{\enorm{x^2(\XX)}}{\gwnorm{x^2}}=\frac{\sqrt{4^2+16^2+4^2}}{\sqrt{24}}=\sqrt{12}\approx 3.464\ .
    \end{align}

    \item Since $\mathcal{O}$ does not possess border terms of degree three, the algorithm terminates with the outputs $\mathcal{O} = [1, x]$ and $G=[g_y, g_{xy}, g_{x^2}]$. The polynomials $g_y$ and $g_{xy}$ are only $\epsilon$-vanishing, while $g_{x^2}$ is exactly vanishing.
    \end{enumerate} 

\end{example}

   Step 2 is the last step that augments $\mathcal{O}$. There, the minimal generalized eigenvalue is 8, so all non-constant, gradient-weighted normalized polynomials supported by the final $\mathcal{O}$ have evaluations greater than $\sqrt{8}\approx 2.8284$. A tolerance greater than this value changes the set of terms. The greatest admissible tolerance for the border polynomials is $\sqrt{2/3}\approx 0.8165$. A tolerance less than this value changes the set of terms.

   The following example uses the same point set, but this time, it is interested in the dependence of the output on the prescribed tolerance.

\begin{example} [\textbf{Small Tolerance}] \label{ex:ABM+GWN-small-tolerance}
We consider again the point set $\XX = \{p_1, p_2, p_3\}$ with
        $p_1 = (2,2)$, 
        $p_2 = (-4,1)$, and
        $p_3 = (2, 0)$.
First, we run the ABM+GWN algorithm with tolerance $\epsilon=0$ and the graded reverse lexicographic term ordering $\sigma$, i.e., $y<x$. Since we only consider exact vanishing, $g(\XX)=0$, we do not have to worry about normalizing the border polynomials. In this exact setting, of course, it would be more efficient simply to compute the kernel of $M$, but for comparison with other tolerances, we still consider the generalized eigenvalue problem.
\begin{enumerate}
    \item As in Example~\ref{ex:ABM+GWN-middle-tolerance}, the algorithm starts with $\mathcal{O}= [1]$, $T_1=[y,x]$, $G=[\,]$, and trial term $b= y$. 
    \begin{align}
        M_y=\begin{pmatrix}
            1 & 2 \\ 
            1 & 1   \\ 
            1 & 0 
        \end{pmatrix}
        \ ,\quad M_y^{\intercal}M_y = \begin{pmatrix} 3 & 3 \\ 3 & 5\end{pmatrix}
        \quad\text{and}\quad
        D_{y}^2=\begin{pmatrix}
            0 & 0 \\ 0 & 3
        \end{pmatrix}\ .
    \end{align}
    The generalized eigenvalues are 2/3 and $\infty$. There is no generalized eigenvalue zero, so there is no non-zero polynomial supported by $\mathcal{O}\cup\{y\}$ that is exactly vanishing on $\XX$. We append $y$ to $\mathcal{O}$ and, already, the list of terms deviates from Example~\ref{ex:ABM+GWN-middle-tolerance}.
    
  \item  Now $\mathcal{O}= [1,y]$, $T_1=[y,x]$, $G=[\,]$, and trial term $b= x$. The matrices
\begin{align}
       M_{x} = \begin{pmatrix}
            1 & 2 & \phantom{-}2\\ 
            1 & 1 & -4\\
           1 & 0 & \phantom{-}2
           \end{pmatrix}
            \ ,\quad M_y^{\intercal}M_x = \begin{pmatrix} 3 & 3 & 0 \\ 3 & 5 & 0 \\ 0 & 0 & 24\end{pmatrix}
        \quad\text{and}\quad
        D_{x}^2=\begin{pmatrix}
            0 & 0 & 0 \\ 0 & 3 & 0 \\ 0 & 0 & 3
        \end{pmatrix}\ .
   \end{align}
   lead to the generalized eigenvalues $2/3$, $8$, and $\infty$. We append $x$ to $\mathcal{O}$ and move to the next degree.

   \item Now $\mathcal{O}= [1,y,x]$, $T_2=[y^2,xy, x^2]$, $G=[\,]$, and trial term $b= y^2$. We need the gradient-weighted semi-norm of the trial term.
   \begin{align}
       \gwnorm{y^2}=\frac{1}{2}\enorm{\nabla y^2(\XX)}=\frac{1}{2}\enorm{\begin{pmatrix}
           0(\XX) \\ 2y(\XX)
       \end{pmatrix}}
       =\sqrt{5}\ .
   \end{align}
   Then
\begin{align}
       M_{y^2} = \begin{pmatrix}
            1 & 2 & \phantom{-}2 & 4\\ 
            1 & 1 & -4 & 1\\
           1 & 0 & \phantom{-}2 & 0
           \end{pmatrix}
           \ ,\quad M_{y^2}^{\intercal}M_{y^2} = \begin{pmatrix} 3 & 3 & 0 & 5\\ 3 & 5 & 0 & 9\\ 0 & 0 & 24& 4\\ 5 & 9 & 4 & 17\end{pmatrix}
        \quad\text{and}\quad
        D_{y^2}^2=\begin{pmatrix}
            0 & 0 & 0 & 0\\ 0 & 3 & 0 &0\\ 0 & 0 & 3& 0\\ 0 & 0 & 0 & 5
        \end{pmatrix}\ .
   \end{align}
   The generalized eigenvalues are 0, 2.2305, 8.1695, $\infty$. The zero eigenvalue has the gradient-weighted normalized eigenvector $g_{y^2}=0.2419y^2-0.0403x-0.4839y+0.0806$, which up to roundoff error is a multiple of $6y^2-x-12y+2$.

   \item Now $\mathcal{O}= [1,y,x]$, $T_2=[y^2,xy, x^2]$, $G=[g_{y^2}]$, and trial term $b= xy$. The kernel of
\begin{align}
       M_{xy} = \begin{pmatrix}
            1 & 2 & \phantom{-}2 & \phantom{-}4\\ 
            1 & 1 & -4 & -4\\
           1 & 0 & \phantom{-}2 & \phantom{-}0
           \end{pmatrix}
           \ ,\quad M_{xy}^{\intercal}M_{xy} = \begin{pmatrix} 3 & 3 & 0 & 0\\ 3 & 5 & 0 & 4\\ 0 & 0 & 24& 24\\ 0 & 4 & 24 & 32\end{pmatrix}
        \quad\text{and}\quad
        D_{xy}^2=\begin{pmatrix}
            0 & 0 & 0 & 0\\ 0 & 3 & 0 &0\\ 0 & 0 & 3& 0\\ 0 & 0 & 0 & 14.5
        \end{pmatrix}\ .
   \end{align}
   The generalized eigenvalues are 0, 1.1115, 9.7621, $\infty$. The zero eigenvalue has the gradient-weighted normalized eigenvector $g_{xy}=0.1841xy-0.1841x-0.3682y+0.3682$, which up to roundoff error is a multiple of $xy-x-2y+2=(y-1)(x-2)$.

    \item Now $\mathcal{O}= [1,y,x]$, $T_2=[y^2,xy, x^2]$, $G=[g_{y^2}, g_{xy}]$, and trial term $b= x^2$. The kernel of
\begin{align}
       M_{x^2} = \begin{pmatrix}
            1 & 2 & \phantom{-}2 & \phantom{1}4\\ 
            1 & 1 & -4 & 16\\
           1 & 0 & \phantom{-}2 & \phantom{1}4
           \end{pmatrix}
           \ ,\quad M_{x^2}^{\intercal}M_{x^2} = \begin{pmatrix} 3 & 3 & 0 & 24\\ 3 & 5 & 0 & 24\\ 0 & 0 & 24& -48\\ 24 & 24 & -48 & 288\end{pmatrix}
        \quad\text{and}\quad
        D_{x^2}^2=\begin{pmatrix}
            0 & 0 & 0 & 0\\ 0 & 3 & 0 &0\\ 0 & 0 & 3& 0\\ 0 & 0 & 0 & 24
        \end{pmatrix}\ .
   \end{align}
   The generalized eigenvalues are 0, 0.6667, 12.0000, $\infty$. The zero eigenvalue has the gradient-weighted normalized eigenvector $g_{x^2}=0.1667x^2+0.3333x-0.0000y-1.3333$, which up to roundoff error is a multiple of $x^2+2x-8=(x-2)(x+4)$.
   \end{enumerate}
\end{example}

All decisions and results in Example~\ref{ex:ABM+GWN-small-tolerance} will be the same if the tolerance satisfies $0\le \epsilon < \sqrt{2/3}$. For a tolerance $\sqrt{2/3}\le \epsilon<\sqrt{8}$ the decisions and results follow Example~\ref{ex:ABM+GWN-middle-tolerance}. Finally, a tolerance $\sqrt{8}\le \epsilon$ leads to an approximately vanishing ideal of $\XX$ which is the exact vanishing ideal of the single point $(0, 1)$:

\begin{example} [\textbf{Large Tolerance}]\label{ex:ABN+GWN-large-tolerance}
    Let $\epsilon\ge \sqrt{8}$ and again $\XX = \{p_1, p_2, p_3\}$ with
        $p_1 = (2,2)$, 
        $p_2 = (-4,1)$, and
        $p_3 = (2, 0)$.
We run the ABM+GWN algorithm again with the graded reverse lexicographic term ordering $\sigma$.
    \begin{enumerate}[(i)]
        \item This step is the same as in Example~\ref{ex:ABM+GWN-middle-tolerance}. The polynomial $g_y=(y-1)/\sqrt{3}$ is appended to $G$
        \item The next step uses $\mathcal{O} = [1]$, $T_1=[y,x]$, $G=[g_y]$, and trial term $b= x$. As in Example~\ref{ex:ABM+GWN-middle-tolerance}
    \begin{align}
        M_x^{\intercal}M_x = \mqty(3 & 0 \\ 0 & 24)\quad\text{and}\quad D_x^2=\begin{pmatrix}
        \gwnorm{1}^2 & 0 \\
        0 & \gwnorm{x}^2
    \end{pmatrix}
    =\mqty(0 & 0 \\ 0 & 3)\ .
    \end{align}
    The generalized eigenvalues are 8 and $\infty$. The generalized eigenvector belonging to eigenvalue 8 is $(0,1)^\intercal$. We append $g_x=x/\sqrt{3}$ to $G$. This finishes the trial terms of degree one. The algorithm terminates as there are no degree-two terms in $\partial\mathcal{O}$.
    \end{enumerate}
\end{example}

These three examples show that different tolerances produce quantized levels of different approximately vanishing ideals.

\begin{remark}[\textbf{Computational complexity}]\label{rem:time-complexity}
The additional computational cost that is incurred by the implementation of gradient-weighted normalization is small. The replacement of standard eigenvalue problem of $M^{\intercal}M$ to generalized version of ($M^{\intercal}M$, $D^2$) in Eq.~\eqref{eq:gep} does not change the original complexity $O(|\mathcal{O}|^3)$. The computational cost of $D$ is $O(n\cdot|\XX|\cdot E_n)$, where $E_n$ denotes the cost of evaluating an $n$-variate term at a point. Note that with careful implementation, most of the entries in $M$ and $D$ were computed in previous steps, which even reduces the complexity.. 
\end{remark}

\subsection{Counterexample - ABM with Coefficient Normalization Need Not Return an Order Ideal.}
\label{subsec:counterexample}

For the tolerance $\epsilon=1$, we compute again a basis of an $\epsilon$-vanishing ideal of the point set $\{(2,2),(-4,1),(2,0)\}$. Only this time, we replace the gradient-weighted normalization that is used in Example~\ref{ex:ABM+GWN-middle-tolerance} by coefficient normalization. Thus, we compute the result of the original ABM algorithm. Computationally, this means that the generalized eigenvalue problems are replaced by the corresponding ordinary eigenvalue problems. In other words, the diagonal matrix $D$ becomes the identity matrix.

The output is going to be a set of terms $\mathcal{O}$ that is not an order ideal. Example~\ref{ex:ABM+GWN-small-tolerance} for $0\le \epsilon< \sqrt{2/3}$, Example~\ref{ex:ABM+GWN-middle-tolerance} for $\sqrt{2/3} \le \epsilon<\sqrt{8}$, and Example~\ref{ex:ABN+GWN-large-tolerance} for $ \sqrt{8}\le\epsilon$ showed that for all tolerances, the ABM algorithm with gradient-weight normalization computed an order ideal for this set of points.

\begin{example}\label{ex:not-an-order-ideal}
    Let $\XX = \{p_1, p_2, p_3\}$ with
    \begin{align}
        p_1 = (2,2), \quad
        p_2 = (-4,1), \quad
        p_3 = (2, 0) \ .
    \end{align}
   We run the ABM algorithm with coefficient normalization, with tolerance $\epsilon=1$ and the graded reverse lexicographic term ordering $\sigma$, i.e., $y<x$. In this setting, the ABM algorithm returns a term set $\mathcal{O}$ that is not an order ideal. Here are the steps of the algorithm:

   \begin{enumerate}
       \item The algorithm starts with $\mathcal{O} = [1]$, $\partial\mathcal{O}=\{y,x\}$, and trial term $b= y$. Thus it computes 
    \begin{align}
        M_y = \mqty(1 & 2 \\ 
                    1 & 1   \\ 
                    1 & 0 
                    ) \quad\text{ and }\quad M_y^{\intercal}M_y = \mqty(3 & 3 \\ 3 & 5).
    \end{align}
    The eigenvalues are $\lambda_{1,2}=4\pm\sqrt{10}$, so $\lambda_{\min}\approx 0.8377 < \epsilon^2$. The eigenvector for the minimal eigenvalue produces 
    the coefficient-normalized basis polynomial $g_y=0.5847y-0.8112$. 
    
    \item The next step uses $\mathcal{O} = [1]$, $\partial\mathcal{O}=\{y,x\}$, $G=[g_y]$, and trial term $b= x$. Therefore 
    \begin{align}
        M_x = \mqty(1 & \phantom{-}2 \\ 
                    1 & -4 \\
                    1 & \phantom{-}2 ) \quad\text{ and }\quad M_x^{\intercal}M_x = \mqty(3 & 0 \\ 0 & 24).
    \end{align}
    Since all eigenvalues, 3 and 24, are greater than $\epsilon^2$, the set $\mathcal{O}$ is appended to $\mathcal{O} = [1, x]$. 
    
    \item Now $\mathcal{O} = [1, x]$, $\partial\mathcal{O}=\{xy, x^2\}$, $G=[g_y]$, and with trial term $b=xy$. We have
    \begin{align}
        M_{xy} = \mqty(1 & \phantom{-}2 & \phantom{-}4\\ 
                       1 & -4 & -4\\
                       1 & \phantom{-}2 & \phantom{-}0
                       ) \quad\text{ and }\quad M_{xy}^{\intercal}M_{xy} 
              = \mqty(3 & 0 & 0  \\ 
                      0 & 24 & 24 \\
                      0 & 24 & 32
                      )
    \end{align}
    with eigenvalues 3, 3.6689, and 52.3311. Since they are all larger than $\epsilon^2$, the list of terms becomes $\mathcal{O} = [1, x, xy]$, which is not an order ideal.

    \item Now $\mathcal{O} = [1, x, xy]$, $\partial\mathcal{O}=\{x^2, xy^2, x^2y\}$, $G=[g_y]$, and $b=x^2$. We have
    \begin{align}
        M_{x^2} = \mqty(1 & \phantom{-}2 & \phantom{-}4 & 4\\ 
                       1 & -4 & -4 & 16\\
                       1 & \phantom{-}2 & \phantom{-}0 & 4
                       ) \quad\text{ and }\quad M_{x^2}^{\intercal}M_{x^2} 
              = \mqty(3 & 0 & 0  & 24\\ 
                      0 & 24 & 24 & -48\\
                      0 & 24 & 32 & -48 \\
                      24 & -48 & -48 & 288)
    \end{align}
    with Matlab-computed eigenvalues -0.000, 3.6402, 35.4599, and 307.8999. The resulting border polynomial is $g_{x^2}=0.1204x^2+0.000xy+0.2408x-0.9631$. It is appended to $G$. It is the coefficient-normalized scalar multiple of the exactly vanishing polynomial $h_{x^2}=(x-2)(x+4)=x^2+2x-8$ with scalar $1/\sqrt{1^2+2^2+8^2}\approx 0.1204$. 

    \item Now $\mathcal{O} = [1, x, xy]$, $\partial\mathcal{O}=\{x^2, xy^2, x^2y\}$, $G=[g_y, g_{x^2}]$, and $b=xy^2$. We have
    \begin{align}
        M_{xy^2} = \mqty(1 & \phantom{-}2 & \phantom{-}4 & \phantom{-}8\\ 
                       1 & -4 & -4 & -4\\
                       1 & \phantom{-}2 & \phantom{-}0 & \phantom{-}0
                       ) \quad\text{ and }\quad M_{xy^2}^{\intercal}M_{xy^2} 
              = \mqty(3 & 0 & 0  & 4\\ 
                      0 & 24 & 24 & 32\\
                      0 & 24 & 32 & 48 \\
                      4 & 32 & 48 & 80)
    \end{align}
    with Matlab-computed eigenvalues -0.000, 3.1948, 10.3819, and 125.4233. The resulting coefficient-normalized polynomial is $g_{xy^2}=0.3721xy^2-0.7442xy+0.2481y-0.4961$, which is appended to $G$.
    
    \item Finally $\mathcal{O} = [1, x, xy]$, $\partial\mathcal{O}=\{x^2, xy^2, x^2y\}$, $G=[g_y, g_{x^2},g_{xy^2}]$, and $b=x^2y$. 
    \begin{align}
        M_{x^2y} = \mqty(1 & \phantom{-}2 & \phantom{-}4 & \phantom{-}8\\ 
                       1 & -4 & -4 & 16\\
                       1 & \phantom{-}2 & \phantom{-}0 & \phantom{-}0
                       ) \quad\text{ and }\quad M_{x^2y}^{\intercal}M_{x^2y} 
              = \mqty(3 & \phantom{-}0 & \phantom{-}0  & \phantom{-}24\\ 
                      0 & \phantom{-}24 & \phantom{-}24 & -32\\
                      0 & \phantom{-}24 & \phantom{-}32 & -48 \\
                      24 & -32 & -48 & 320)
    \end{align}
    with Matlab-computed eigenvalues -0.000, 3.4909, 42.0609, and 333.4482. The resulting coefficient-normalized polynomial is $g_{x^2y}=0.1085x^2y-0.2169xy+0.4339x-0.8677$, which is appended to $G$. The corresponding exactly vanishing polynomial is obtained from $y\cdot h_{x^2}=x^2y+2xy-8y$ which is reduced by $4\cdot h_{xy}=4(xy-2y-x+2)$. Then normalize $y\cdot h_{x^2}-4\cdot h_{xy}=x^2y-2xy+4x-8$ with respect to its coefficient norm.

    \item Since $\mathcal{O}$ does not contain border terms of degree four, the algorithm terminates with the outputs $\mathcal{O} = [1, x, xy]$ and $G=[g_y, g_{x^2},g_{xy^2},g_{x^2y}]$.
    \end{enumerate} 
\end{example}

This example produces a set of terms $\mathcal{O}$ that is not an order ideal. Anyway, $|\mathcal{O}|=|\XX|$, so $\mathcal{O}$ has maximal size.

\section{Features of Gradient-Weighted Normalization}
\label{sec:gradient-weighted-features}

Gradient-weighted normalized bases of approximately vanishing ideals have two key features related to changes of the point set $\XX$. These bases are robust against perturbations of $\XX$, and the bases scale nicely when $\XX$ is scaled.

\subsection{Robustness against Perturbations of the Data Points}

\begin{proposition}\label{prop:gradient-norm-upper-bound}
    Let $\mathbb{X}\subset\mathbb{R}^n$ be a non-empty, finite set of points.
    For every polynomial $g\in\mathbb{R}[x_1,\dots,x_n]$, we have
    \begin{align}
        \enorm{\nabla g(\XX)}\le \degree{g}\cdot \gwnorm{g}\ .
    \end{align}
    In particular, every gradient-weight normalized polynomial satisfies
    \begin{align}
        \enorm{\nabla g(\XX)}\le \degree{g}\ .
    \end{align}
\end{proposition}

\begin{proof}
    For a constant polynomial $g$, both $\enorm{\nabla g(\XX)}$ and the gradient-weighted semi-norm are zero, so equality holds. Next, let $g=\sum_{t\in\supp{g}} c_t\cdot t$ be non-constant, so $\deg g\ne 0$. For each term $t\in\supp{g}$, the Euclidean degree satisfies $D(t)=\sqrt{\sum_k {\degreek{k}{t}^2}}\le \degree{t} \le \degree{g}$. Then 
    \begin{align}
    \frac{\enorm{\nabla g(\XX)}}{\degree{g}}
    & =\sqrt{\sum_{p\in\XX}\enorm{\nabla g(p)}^2}/\degree{g}
    \le \sqrt{\sum_{p\in\XX}\sum_{t\ne 1} c_t^2 \enorm{\nabla t(p)}^2/\degree{g}^2}\\
    &\le \sqrt{\sum_{p\in\XX}\sum_{t\ne 1} c_t^2 \enorm{\nabla t(p)}^2/D_t^2}
    =\gwnorm{g}
\end{align}
\end{proof}

\begin{proposition}\label{prop:perturbation}
    Let $\XX = \{p_1,\ldots, p_m\}\subset\RR^n$ and $\epsilon \ge 0$. Let $(\mathcal{O}, G)\subset \TT\times \PP$ be the output of the ABM+GWN algorithm. Let $\Delta=\{\delta p_1,\ldots, \delta p_m\}\subset \RR^n$ be perturbations for each point. By computation of the algorithm
    each $g \in G$ is gradient-weighted normalized and $\epsilon$-vanishing on $\XX$. We have
        \begin{align}
        \enorm{g(\XX+\Delta)-g(\XX)}\le \degree{g}\cdot \norm{\Delta}_{\max}+o(\norm{\Delta}_{\max})
    \end{align}
    where $\|\Delta\|_{\max} = \max_{1\le i\le m} \enorm{\delta p_i}$, and $o(\cdot)$ denotes Landau's small o. In particular,
    \begin{align}
        \enorm{g(\XX+\Delta)} \le \epsilon + \degree{g}\cdot\norm{\Delta}_{\max}  + o(\norm{\Delta}_{\max})\ ,
    \end{align}
    so if $g$ is strictly $\epsilon$-vanishing on $\XX$, i.e., $\enorm{g(\XX)}<\epsilon$, then, for sufficiently small perturbations, $g$ will be $\epsilon$-vanishing on $\XX+\Delta$. 
\end{proposition}
\begin{proof} 
    Apply the Taylor expansion in each coordinate
    \begin{align}
    \begin{pmatrix}
        g(p_1+\delta p_1)\\
        \vdots\\
        g(p_m+\delta p_m)
    \end{pmatrix}
        =
    \begin{pmatrix}
        g(p_1)\\
        \vdots\\
        g(p_m)
    \end{pmatrix}  
    +\begin{pmatrix}
        \nabla g(p_1)\cdot \delta p_1\\
        \vdots\\
        \nabla g(p_m)\cdot \delta p_m 
    \end{pmatrix}+
    \begin{pmatrix}
    o\bigl(\enorm{\delta p_1}\bigr)\\
    \vdots\\
    o\bigl(\enorm{\delta p_m}\bigr)
    \end{pmatrix}\quad.   
    \end{align}
    The Cauchy-Schwarz inequality leads to
    \begin{align}
    \begin{pmatrix}
        g(p_1+\delta p_1)\\
        \vdots\\
        g(p_m+\delta p_m)
    \end{pmatrix}
        =
    \begin{pmatrix}
        g(p_1)\\
        \vdots\\
        g(p_m)
    \end{pmatrix}  
    +\begin{pmatrix}
        \enorm{\nabla g(p_1)}\cdot \enorm{\delta p_1}\\
        \vdots\\
        \enorm{\nabla g(p_m)}\cdot \enorm{\delta p_m}
    \end{pmatrix}+
    \begin{pmatrix}
    o\bigl(\enorm{\delta p_1}\bigr)\\
    \vdots\\
    o\bigl(\enorm{\delta p_m}\bigr)
    \end{pmatrix}\quad.
    \end{align}
    We obtain for the Euclidean norm
    \begin{align}
    \enorm{g(\XX+\Delta)}
        &=
    \enorm{g(\XX)}
    +\sqrt{\sum_{i=1}^m
        \enorm{\nabla g(p_i)}^2\cdot \enorm{\delta p_i}^2}+
    o\bigl(\|{\Delta}\|_{\max}\bigr)\\
    &\le
    \epsilon
    +\|\Delta\|_{\max}\cdot\sqrt{\sum_{i=1}^m\enorm{\nabla g(p_i)}^2}+
    o\bigl(\|{\Delta}\|_{\max}\bigr)\\
    &=
    \epsilon
    +\|\Delta\|_{\max}\cdot\enorm{\nabla g(\XX)}+
    o\bigl(\|{\Delta}\|_{\max}\bigr)\quad.
    \end{align}
    Then apply Proposition~\ref{prop:gradient-norm-upper-bound}.
\end{proof}
Proposition~\ref{prop:perturbation} connects the change in algebraic distance (i.e., $\norm{g(\XX+\Delta) - g(\XX))}$) to that in geometric distance (i.e., $\norm{\Delta}_{\max}$). Importantly, our gradient-weighted normalization removes the dependency on $\XX$ from the right-hand side of the inequalities. If one uses coefficient normalization, the second term becomes $\degree{g}\cdot\norm{\Delta}_{\max} \cdot R$ with some $R \ge 0$ depending on $\mathbb{X}$; thus, the impact of the perturbations on the extent of vanishing is unclear.

\subsection{Gradient-Weighted Normalization Produces Invariance with Respect to Data Point Scaling}

The key feature of implementing gradient-weighted normalization is that the structure of the computed approximate basis of an approximately vanishing ideal is invariant of the scale of the points and tolerance. More precisely, the set of terms, its border, and the supports of the basis polynomials are invariant. On the other hand, an example is going to show that the computation with coefficient normalization may be affected by the scale of the points and tolerance.

\begin{theorem}\label{thm:scaling-consistency}
Let $\XX\subset\RR^n$ be a set of points, $\epsilon \ge 0$ a tolerance, and $\sigma$ a degree-compatible term ordering. Moreover, let $\alpha > 0$ be a scaling factor. On the one hand, apply the ABM+GWN algorithm to $\XX$ with tolerance $\epsilon$ to obtain the result  $(\mathcal{O}, G)$. On the other hand, apply the ABM+GWN algorithm to the scaled points $\alpha\cdot \XX$ with the scaled tolerance $(\alpha\cdot\epsilon)$ to obtain $({\mathcal{O}_\alpha}, G_\alpha)$. 

The ABM+GWN algorithm is scale-invariant in the following sense. 
\begin{enumerate}[(i)]
\item The set of terms are the same, $\mathcal{O}_\alpha=\mathcal{O}$.
\item The polynomials in $G_\alpha=[g_{\alpha,1},\dots,g_{\alpha,r}]$ are coefficient-wise scaled versions of those in $G=[g_1,\dots,g_r]$, namely the coefficients $c_{\alpha,i,t}$ and $c_{i,t}$ of a term $t$ are transformed into one another via $\alpha^{\deg t-1}\cdot c_{\alpha,i,t}=c_{i,t}$. This transformation rule only depends on the degree of $t$ and is independent of the polynomial index $i$. This relationship can be expressed invariantly as  
\begin{align}
    g_{\alpha,i}(\alpha\cdot \mathbf{p})=\alpha\cdot g_{i}\bigl(\alpha^{-1}\cdot(\alpha\cdot \mathbf{p})\bigr) \quad\text{for}\ 1\le i\le r\ . 
\end{align}
\end{enumerate}
\end{theorem}

\begin{proof}
    We show inductively that the steps of the ABM+GWN algorithm on the scaled data correspond to those on the original data. Both start with the lists $\mathcal{O}=\mathcal{O}_\alpha=[1]$ and $G=G_\alpha=[]$. Since both runs of the algorithm use the same term ordering $\sigma$, the loops are performed in the same order. 

    Now consider an arbitrary step that starts with $\mathcal{O}_\alpha=[t_1,\dots, t_s]=\mathcal{O}$ and with the $\sigma$-dependent, but scale-indepedent trial term $b\in\partial\mathcal{O}$. The evaluation matrices for the different scales are related in the following way.
    \begin{align}
        M_\alpha &=\begin{pmatrix}
             t_1(\alpha\cdot\XX) &\dots & t_s(\alpha\cdot\XX) & b(\alpha\cdot\XX) 
        \end{pmatrix}
        =\begin{pmatrix}
             \alpha^{\deg t_1}t_1(\XX) &\dots & \alpha^{\deg t_s}t_s(\XX) & \alpha^{\deg b}b(\XX) 
        \end{pmatrix}\\
        &=\begin{pmatrix}
             t_1(\XX) &\dots & t_s(\XX) & b(\XX) 
        \end{pmatrix}\cdot A= M\cdot A\ ,
    \end{align}
    where $A$ is the diagonal scaling matrix $\text{diag}(\alpha^{\deg t_1},\dots, \alpha^{\deg t_s},\alpha^{\deg b})$, which is invertible. Hence, the left-hand side of the generalized eigenvalue problem of the scaled points uses the matrix $M_\alpha^\intercal M_\alpha=A^\intercal M^\intercal M A$.

    For the right-hand side, the concept  $\|t\|_{\nabla,\alpha\XX}$ is notationally tricky. The gradient-weighted semi-norm requires taking the gradient first and then plugging in the scaled points---the scaling factor does not lead to using the chain rule. So, the scaled points are plugged into differentiated terms of degree $\deg t-1$. Therefore, we replace the possibly misleading notation $\enorm{\nabla t(\alpha\cdot\XX)}$ by $\enorm{\nabla t\big\rvert_{\alpha\cdot \XX}}$ and obtain 
    \begin{align}
        \|t\|_{\nabla,\alpha\XX}=\frac{\enorm{\nabla t\big\rvert_{\alpha\cdot \XX}}}{D(t)}= \alpha^{\deg t-1}\enorm{\nabla t(\XX)}/D(t)=\alpha^{\deg t-1}\cdot \gwnorm{t}\ .
    \end{align}
    Hence, the diagonal matrix of the scaled points relates to the one of the original points as $D_\alpha=D\cdot \frac{1}{\alpha}A$. Therefore, the scaled version of the generalized eigenvalue problem is
    \begin{align}
        M^\intercal_\alpha M_\alpha \mathbf{v}_\alpha = \lambda_\alpha D^\intercal_\alpha D_\alpha \mathbf{v}_\alpha 
        \ \Leftrightarrow\  A^\intercal M^\intercal M A\mathbf{v}_\alpha=\lambda_\alpha \frac{1}{\alpha}A^\intercal D^\intercal D\frac{1}{\alpha}A\mathbf{v}_\alpha
        \ \Leftrightarrow\  M^\intercal M (A\mathbf{v}_\alpha)=\frac{\lambda_\alpha}{\alpha^2}D^\intercal D(A\mathbf{v}_\alpha)
        \ .
    \end{align}
     Hence the generalized eigenvalues of the scaled data points and the original data points are related by $\lambda_{\alpha}/\alpha^2=\lambda$, and the eigenvectors are related via the invertible transformation $\mathbf{v}_\alpha=A^{-1}\mathbf{v}$. The respective normalization conditions are
    \begin{align}
        \text{(original)}\quad \mathbf{v}^\intercal D^\intercal D\mathbf{v}=1 \qquad\text{and}\quad \text{(scaled)}\quad
        \mathbf{v}_\alpha^\intercal D_\alpha^\intercal D_{\alpha}\mathbf{v}_\alpha=1
        \ \Leftrightarrow\ \mathbf{v}_\alpha^\intercal A^\intercal D^\intercal D A\mathbf{v}_\alpha=1\ .
    \end{align}
    
    The decision $\sqrt{\lambda_{\min}}\le \epsilon$ for the original data points is equivalent to the one $\sqrt{\lambda_{\alpha,{\min}}}=\alpha\cdot\sqrt{\lambda_{\min}}\le \alpha\cdot\epsilon$ for the scaled data points. Hence, the decision of whether the trial term is appended to the list of terms or becomes the border term of an $\epsilon$-vanishing polynomial is independent of the data point scaling.
    
    The transformation of the generalized eigenvectors $A\mathbf{v}_\alpha=\mathbf{v}$ gives the scaling of the polynomial coefficients.
\end{proof}

\begin{example}\label{example:scaling-consistency}
We started with six equiangular points on the unit circle, $\tilde{\XX} = \{(\cos k\pi/3, \sin k\pi/3)\mid 1\le k\le 6\}$, perturbed them with additive two-dimensional Gaussian noise $\mathcal{N}(\bigl(0,0), 0.1\cdot I_2\bigr)$, and rounded the coordinates to hundredths to create the point set
\begin{align}
    \XX = \{( 0.41,  0.86), (-0.50,  0.85), (-1.12, -0.05), (-0.55, -0.97), ( 0.44, -0.91), ( 1.01, -0.02)\}\ .
\end{align}
We apply the ABM+GWN algorithm to $\XX$, with tolerance $\epsilon=0.1$, and the graded reverse lexicographic order $\sigma$. The computed order ideal is $\mathcal{O} = \{1, x, y, x^2, xy, x^3\}$. The computed basis $G$ consists of one quadratic, two degree-four polynomials, and two cubic polynomials, 
\begin{align}
    g_1 &= 1.0 x_{0}^{2} - 0.09411 x_{0} x_{1} + 1.126 x_{1}^{2} + 0.1173 x_{0} + 0.1122 x_{1} - 1.126, \quad with\quad \enorm{g_1} = 0.028, \\
    g_2 &= 1.0 x_{0}^{3} - 255.6 x_{0}^{2} x_{1} - 3.994 x_{0}^{2} - 25.81 x_{0} x_{1} + 1.766 x_{0} + 57.28 x_{1} - 3.33, \quad with\quad \enorm{g_2} = 0.0, \\
    g_3 &= 1.0 x_{0}^{3} + 1.125 x_{0} x_{1}^{2} + 0.09809 x_{0}^{2} + 0.1216 x_{0} x_{1} - 1.13 x_{0} - 0.02142 x_{1} + 0.01206, \quad with\quad \enorm{g_3} = 0.0, \\
    g_4 &= 1.0 x_{0}^{4} + 0.209 x_{0}^{3} - 1.343 x_{0}^{2} + 0.009772 x_{0} x_{1} - 0.136 x_{0} - 0.01868 x_{1} + 0.2515, \quad with\quad \enorm{g_4} = 0.0, \\
    g_5 &= 1.0 x_{0}^{3} x_{1} + 0.01684 x_{0}^{3} - 0.01374 x_{0}^{2} - 0.2342 x_{0} x_{1} + 0.01319 x_{0} + 0.02255 x_{1} - 0.0003311, \quad with\quad \enorm{g_5} = 0.0.
\end{align}
Particularly, the zero set of $g_1$ is close to the unit circle. 

For the scaled inputs $(\alpha\cdot \XX, \alpha\cdot\epsilon)$ with $\alpha = 0.1$, the ABM+GWN algorithm computes $G_\alpha$, which has the same configuration as $G$. The order ideal is unchanged, i.e., $\mathcal{O}_{\alpha} = \mathcal{O}$. For example, Let $g_{i,\alpha} \in G_\alpha$ be the quadratic polynomial corresponding to $g_i$, 
\begin{align}
    g_{1,\alpha} &= 1.0 x_{0}^{2} - 0.0941 x_{0} x_{1} + 1.126 x_{1}^{2} + 0.01173 x_{0} + 0.01122 x_{1} - 0.01126, \quad with\quad \enorm{g_{1,\alpha}} = 0.0028, \\
    g_{2,\alpha} &= 1.0 x_{0}^{3} - 255.6 x_{0}^{2} x_{1} - 0.3994 x_{0}^{2} - 2.581 x_{0} x_{1} + 0.01766 x_{0} + 0.5728 x_{1} - 0.00333, \quad with\quad \enorm{g_{2,\alpha}} = 0.0, \\
    g_{3,\alpha} &= 1.0 x_{0}^{3} + 1.125 x_{0} x_{1}^{2} + 0.009809 x_{0}^{2} + 0.01216 x_{0} x_{1} - 0.0113 x_{0} - 0.0002142 x_{1}, \quad with\quad \enorm{g_{3,\alpha}} = 0.0, \\
    g_{4,\alpha} &= 1.0 x_{0}^{4} + 0.0209 x_{0}^{3} - 0.01343 x_{0}^{2} + 9.772 \cdot 10^{-5} x_{0} x_{1} - 0.000136 x_{0} - 1.868 \cdot 10^{-5} x_{1}, \quad with\quad \enorm{g_{4,\alpha}} = 0.0, \\
    g_{5,\alpha} &= 1.0 x_{0}^{3} x_{1} + 0.001684 x_{0}^{3} - 0.0001374 x_{0}^{2} - 0.002342 x_{0} x_{1} + 1.319 \cdot 10^{-5} x_{0} + 2.255 \cdot 10^{-5} x_{1}, \  with \ \  \enorm{g_{5,\alpha}} = 0.0.
\end{align}
The coefficients of the quadratic, linear, and constant terms in $g_{1, \alpha}$ are $0.1^{-1}$, $0.1^{0}$, and $0.1$ times those in $g_1$, respectively. 

On the other hand, the ABM algorithm with coefficient normalization computes the bases sets with different configurations for $(\XX, \epsilon)$ and $(\alpha\cdot \XX, \alpha\cdot\epsilon)$. The order ideal is $\mathcal{O} = \{1, x, y, x^2, xy, x^3\}$ for the former and $\mathcal{O}_{\alpha} = \{1, x, y, x^2\}$ for the latter.
See Figure~\ref{fig:scaling}.
\end{example}
 \begin{figure}[t]
  \centering
  \includegraphics[width=\linewidth]{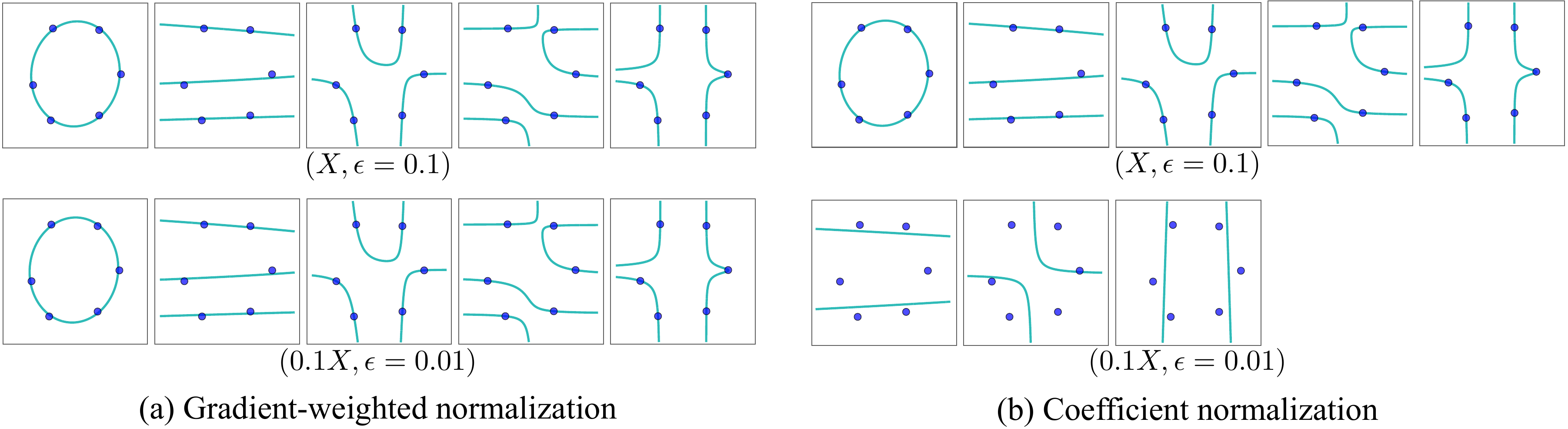}
  \caption{Numerical example of the scaling invariance of gradient-weighted normalization (Theorem~\ref{thm:scaling-consistency}).  (a)~Scaling of the points and tolerance does not have an essential effect on the result of the ABM+GWN algorithm. (b)~In contrast, ABM (with coefficient normalization) sometimes responds differently to scaled data.}
  \label{fig:scaling}
\end{figure}

\section{Coefficient-Normalized Bases Depend on the Data Point Scale}
\label{sec:coefficient-scaling-dependence}

Theorem~\ref{thm:scaling-consistency} shows that gradient-weighted normalization leads to a structural invariance of the computed ideal basis. The approximate bases of the scaled and unscaled point sets correspond to one another. This correspondence breaks down when coefficient normalization is used instead, as illustrated by Example~\ref{example:scaling-consistency}. A reason for this is that coefficient normalization ignores the data points and their scale.

More specifically, a term $t\in \mathcal{O}$---by definition, a term is coefficient-normalized---must have a sufficiently large evaluation at the data points,
\begin{align}
    \enorm{t(\XX)}>\epsilon\ .
\end{align}
For the constant term $t=1$, this establishes an upper bound for the tolerance $\epsilon$ that depends on the size of the point set, $\sqrt{|\XX|}>\epsilon$. For a non-constant term $t$ scaling of the point set leads to
\begin{align}
    \enorm{t(\alpha\cdot \XX)}=\alpha^{\deg t}\cdot\enorm{t(\XX)}\ .
\end{align}
This is already a strong indication that scaling the data points may affect the structure of $\mathcal{O}$. Scaling by a large parameter may enlarge $\mathcal{O}$; scaling by a small parameter may shrink $\mathcal{O}$. This indication is proven in the following setting.

\begin{theorem} \label{thm:failure-scale-invariance}
    Let $\XX \subset\RR^n$ be a non-empty, finite set of points, $\epsilon>0$ a positive tolerance, and $\sigma$ a degree-compatible term ordering. Let $\mathcal{O}$ and $G$ be the output of the ABM algorithm (with coefficient normalization).

    If $\mathcal{O}$ contains an at least quadratic term, then there exists a scaling parameter $\alpha>0$ such that the ABM algorithm with the same term ordering applied to $\alpha\cdot \XX$ with tolerance $\epsilon\cdot \alpha$ produces a different set of terms $\mathcal{O}_\alpha$ and, necessarily, a significantly different set of basis polynomials $G_\alpha$.
\end{theorem}

\begin{proof} Let $t\in\mathcal{O}$ be an at least quadratic term. Then $\enorm{t(\XX)}>\epsilon$, i.e., $\enorm{t(\XX)}=\gamma\cdot\epsilon$ for some $\gamma>1$. Define $0<\alpha<\sqrt[\deg t-1]{1/\gamma}$. Then $\alpha^{\deg t-1}\cdot \gamma<1$ and
\begin{align}
    \enorm{t(\alpha\cdot\XX)}=\alpha^{\deg t}\cdot \enorm{t(\XX)}=\alpha^{\deg t}\cdot \gamma\cdot\epsilon
    < \alpha\cdot\epsilon\ .
\end{align}
Therefore $t\notin \mathcal{O}_\alpha$.
\end{proof}

The scale change need not be large to create a change in the computed set of terms $\mathcal{O}$. The closer $\gamma$ is to 1 and the higher the degrees in $\mathcal{O}$, the closer $\alpha$ will be to 1.

The influence of the scaling parameter on the algebraic structure of bases has hardly been discussed in the literature. \cite{heldt2009approximate} discuss scaling from a numerical perspective and propose to scale the point set into $[-1, 1]^n$. However, as a result of Theorem~\ref{thm:failure-scale-invariance} any scaling in combination with coefficient normalization has algebraic structural consequences. Coefficient normalization should be replaced by gradient-weighted normalization to avoid these consequences.

\section{Application to Perturbed, Scaled Samples -- Numerical Experiments}\label{sec:numerical-experiments}

We now demonstrate that the scaling of data points is a crucial factor in the success of the approximate computation of border bases. 
\begin{figure*}[t]
  \centering
  \includegraphics[width=\linewidth]{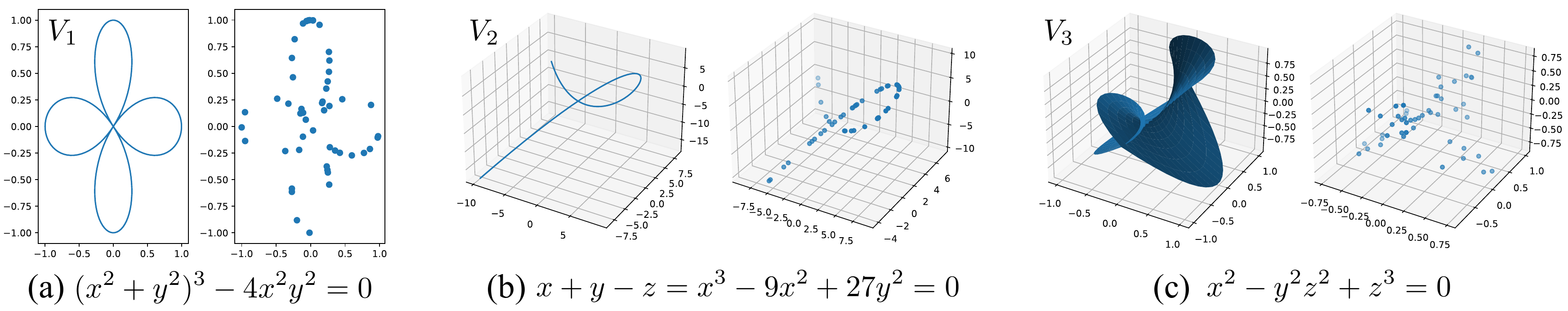}
  \caption{Three affine varieties ($V_1$, $V_2$, and $V_3$) and sampled points (before preprocessing and perturbation).}
  \label{fig:varieties}
\end{figure*}
We construct three datasets from the affine varieties in Figure~\ref{fig:varieties} and examine whether the output of the ABM algorithm---with gradient-weighted or coefficient normalization---is \textit{consistent} across perturbation and scaling. 

\begin{definition}[Scaling Consistency Test] Let $T \in \mathbb{N}$ and $\alpha > 0$. Let $\XX^* \subset \RR^n$ be a finite set of points, and let $\XX$ be its perturbed version. Let $G^*$, $G$, $\widehat{G}$ be the polynomial sets from three runs of the ABM algorithm (with some normalization) with $\XX^*$, $\XX$, and $\alpha\cdot \XX$, respectively. The runs are said to be \textbf{scaling consistent} if the following holds: 
\begin{align}\label{eq:consistency-test}
     |G^* \cap \mathcal{P}_t| = |G \cap  \mathcal{P}_t| = |\widehat{G}_t \cap \mathcal{P}_t|, \quad \text{for all $t = 0, \ldots, T$}.
\end{align}
\end{definition}
Note that the first equality of Eq.~\eqref{eq:consistency-test} requires robustness to perturbations, and the second equality requires that to scaling. For the former, Proposition~\ref{prop:perturbation} suggests robustness of the ABM algorithm with gradient-weighted normalization to perturbations . If the first equality holds, Theorem~\ref{thm:scaling-consistency} on scale invariance guarantees that the second inequality holds as well for any $\alpha > 0$. This is not the case with coefficient normalization, as presented in Theorem~\ref{thm:failure-scale-invariance}.

Datasets were constructed as follows. The three affine varieties in Figure~\ref{fig:varieties} have parametric representations. 
\begin{align}\label{eq:three-varieties}
    V_1 &= \qty{(x,y) \in \mathbb{R}^2 \mid (x^2+y^2)^3 - 4x^2y^2 = 0}, \\
        &= \qty{(x,y) \in \mathbb{R}^2 \mid x = \cos(2u)\cos(u), y = \cos(2u)\sin(u), u\in \mathbb{R}}, \\
    V_2 &= \qty{(x, y, z) \in \mathbb{R}^3 \mid x + y - z = 0, x^3 - 9(x^2 - 3y^2) = 0}, \\
        &= \qty{(x, y, z) \in \mathbb{R}^3 \mid x = 3(3-u^2), y = u(3-u^2), z = x + y, u\in \mathbb{R}}, \\
    V_3 &= \qty{(x, y, z) \in \mathbb{R}^3 \mid x^2-y^2z^2+z^3 = 0}, \\
        &= \qty{(x, y, z) \in \mathbb{R}^3 \mid x = v(u^2-v^2), y=u, z=u^2-v^2, (u, v)\in \mathbb{R}^2}.
\end{align}
From each, we first sampled 50 points from $u \in [-1, 1)$ for $V_1$, $u \in [-2.5, 2.5)$ for $V_2$, and $(u, v) \in [-1, 1)^2$ for $V_3$, respectively. Each sample set was preprocessed by subtracting the mean and scaling to make the average $L_2$ norm of the points unit, which we denote $\XX_i$ for $V_i$ for $i=1,2,3$. The sampled points were then perturbed by an additive Gaussian noise $\mathcal{N}(\mathbf{0}, \nu I)$, where $I$ denotes the identity matrix and $\nu\in \{0.01, 0.05\}$, and then recentered again. The set of such perturbed points from $\XX_i^{*}$ is denoted by $\XX_i$. 

With these datasets, we conducted the Scaling Consistency Test with five scales $\alpha=0.01, 0.1, 1.0, 10, 100$. To remove the dependency on the threshold $\epsilon$, all $\epsilon$ in $[10^{-5}\alpha, \alpha)$ were exhaustively tested for each run with a step size $10^{-3}\alpha$. The value of $T$ was set to $T=6, 3, 3$ for $V_1, V_2, V_3$, respectively, according to the highest degree of polynomials defining them (cf.~Eq.~\eqref{eq:three-varieties}).

Tables~\ref{table:SRT-noise01} and~\ref{table:SRT-noise05} summarize the results, corresponding to the perturbation level $\nu = 0.01, 0.05$, respectively. We performed 20 independent runs for each setup.
We first focus on Table~\ref{table:SRT-noise01}. The success rate column displays the ratio of the cases with successfully passing the Scaling Consistency Test to 20 runs. With gradient-weighted normalization, the ABM algorithm succeeded in all datasets and scales (except $(V_2, \alpha=0.01)$), whereas with coefficient normalization, it succeeded only in specific scales (not necessarily $\alpha=1$). For numerically stable computation, it is common to preprocess data points in a certain range (in our case, mean-zero, unit average $L_2$ norm, and $\alpha=1$). However, our experiment shows that, with coefficient normalization, such preprocessing can lead the approximate border basis construction to fail. In contrast, gradient-weighted normalization provides robustness against such preprocessing. 

Another observation is that the valid range of $\epsilon$ and the extent of vanishing of gradient-weighted normalization both change in proportion to the scale. For example, at $V_2$, the ranges of valid $\epsilon$ change as $(1.94, 2.17)\times 10^{-n}, (n=-3, -2, \ldots, 1)$ and the extent of vanishing changes as $1.80\times 10^{-n}, (n=-3, -2, \ldots, 1)$.  This tendency is supported by the scaling invariance (Theorem~\ref{thm:scaling-consistency}). Note that although the extent of vanishing appears to be large at $\alpha=100$, the signal-to-noise ratio remains unchanged.
With coefficient normalization, the range of $\epsilon$ and the extent of vanishing change in an inconsistent way. At $V_3$, between $\alpha=0.01$ and $\alpha=1.0$, the scale of the range of $\epsilon$ differs by three orders, while between $\alpha=1.0$ and $\alpha=100$ share the same order. For $\alpha=0.1$, the successful case was only $\epsilon=10^{-5}$, where the initial value and the step size of the linear search are $10^{-6}$ and $10^{-4}$, respectively.

When the perturbation level increases to $\nu=0.05$~(Table~\ref{table:SRT-noise05}), similar results were observed; gradient-weighted normalization showed more robustness against a scaling of points than coefficient normalization, and the range of $\epsilon$ and the extent of vanishing changed proportionally to the scaling. 
Besides, coefficient normalization resulted in a lower success rate at $V_2$ and $V_3$. For example, at $(V_2, \alpha=0.1)$, the success rate dropped from 1.00 to 0.00. In contrast, gradient-weighted normalization retained its performance. This result implies the better stability of gradient-weighted normalization against perturbation.

\begin{table*}
    \centering
    \caption{Summary of the Scaling Consistency Test of 20 independent runs with 1\,\% noise. Columns \textit{coeff.} and \textit{grad. w.} present the coefficient and gradient-weighted normalization, respectively. Column \textit{coeff. dist.} presents the distance of the normalized coefficient vectors between the systems from 500 unperturbed and 100 perturbed points, respectively. Column \textit{e.v.} denotes the extent of vanishing at the unperturbed points. The values of the range, coefficient distance, and extent of vanishing are averaged values over 20 independent runs. As indicated by the success rate, the proposed gradient-weighted normalization approach is robust and consistent (see the proportional change in the range and the extent of vanishing) to the scaling, whereas coefficient normalization is not.}\label{table:SRT-noise01}
\begin{tabular}{c|crcccc}
    \toprule
        dataset & normalization & scaling $\alpha$ & range & coeff. dist. & e.v. & success rate \\
    \midrule
    \multirow{10}{*}{$V_1$} & \multirow{5}{*}{coefficient}  & 0.01        & --                             & --         & --                          & 0.00 [00/20]       \\
                       &                              & 0.1         & --                             & --         & --                          & 0.00 [00/20]       \\
                       &                              & 1.0         & [2.28, 2.61] $\times 10^{-2}$  &  1.39      & 1.27 $\times 10^{-2}$       & 1.00 [20/20]       \\
                       &                              & 10          & [1.80, 1.91] $\times 10^{-0}$  &  1.41      & 1.28                        & 1.00 [20/20]       \\
                       &                              & 100         & [1.31, 1.41] $\times 10^{-0}$  &  1.34      & 1.26                        & 0.85 [17/20]       \\
                       \cmidrule{2-7}
                       & \multirow{5}{*}{gradient-weighted}    & 0.01        & --                             &            & --                          & 0.00 [00/20]        \\
                       &                              & 0.1         & [3.33, 4.20] $\times 10^{-3}$  &  1.32      & 4.91 $\times 10^{-3}$       & 1.00 [20/20]       \\
                       &                              & 1.0         & [3.33, 4.20] $\times 10^{-2}$  &  1.32      & 4.91 $\times 10^{-2}$       & 1.00 [20/20]       \\
                       &                              & 10          & [3.33, 4.20] $\times 10^{-1}$  &  1.41      & 4.91 $\times 10^{-1}$       & 1.00 [20/20]       \\
                       &                              & 100         & [3.33, 4.20] $\times 10^{-0}$  &  1.41      & 4.91                        & 1.00 [20/20]       \\
    \midrule
    \multirow{10}{*}{$V_2$} & \multirow{5}{*}{coefficient}  & 0.01        & --                             & --         & --                          & 0.00 [00/20]       \\
                       &                              & 0.1         & --                             & --         & --                          & 0.00 [00/20]       \\
                       &                              & 1.0         & [0.69, 1.30] $\times 10^{-1}$  & 0.0121     & 3.67 $\times 10^{-2}$       & 1.00 [20/20]       \\
                       &                              & 10          & [1.43, 1.63] $\times 10^{-0}$  & 0.577      & 1.28                        & 1.00 [20/20]       \\
                       &                              & 100         & --                             &  --        & --                          & 0.00 [00/20]       \\
                       \cmidrule{2-7}
                       & \multirow{5}{*}{gradient-weighted}    & 0.01        & [1.94, 2.17] $\times 10^{-3}$  &  0.612     & 1.80 $\times 10^{-3}$       & 1.00 [20/20]        \\
                       &                              & 0.1         & [1.94, 2.17] $\times 10^{-2}$  &  0.518     & 1.80 $\times 10^{-2}$       & 1.00 [20/20]        \\
                       &                              & 1.0         & [1.94, 2.17] $\times 10^{-1}$  &  0.104     & 1.80 $\times 10^{-1}$       & 1.00 [20/20]        \\
                       &                              & 10          & [1.94, 2.17] $\times 10^{-0}$  &  0.455     & 1.80                        & 1.00 [20/20]        \\
                       &                              & 100         & [1.94, 2.17] $\times 10^{+1}$  &  0.681     & 1.80 $\times 10^{+1}$       & 1.00 [20/20]        \\
    \midrule
    \multirow{10}{*}{$V_3$} & \multirow{5}{*}{coefficient}  & 0.01        & --                             &  --        & --                          & 0.00 [00/20]       \\
                       &                              & 0.1         & [1.00, 1.00] $\times 10^{-6}$  &  0.836     & 1.93 $\times 10^{-6}$       & 0.65 [13/20]       \\
                       &                              & 1.0         & [7.06, 8.85] $\times 10^{-3}$  &   1.23     & 2.35 $\times 10^{-2}$       & 0.95 [19/20]       \\
                       &                              & 10          & [1.12, 1.26] $\times 10^{-0}$  &   1.20     & 2.92                        & 1.00 [20/20]       \\
                       &                              & 100         & [1.68, 1.72] $\times 10^{-0}$  &   1.05     & 2.18                        & 0.75 [15/20]       \\
                       \cmidrule{2-7}
                       & \multirow{5}{*}{gradient-weighted}    & 0.01        & [1.67, 2.36] $\times 10^{-4}$  &   1.30     & 5.97 $\times 10^{-4}$       & 1.00 [20/20]        \\
                       &                              & 0.1         & [1.67, 2.36] $\times 10^{-3}$  &   1.29     & 5.97 $\times 10^{-3}$       & 1.00 [20/20]        \\
                       &                              & 1.0         & [1.67, 2.36] $\times 10^{-2}$  &   1.29     & 5.97 $\times 10^{-2}$       & 1.00 [20/20]        \\
                       &                              & 10          & [1.67, 2.36] $\times 10^{-1}$  &   1.15     & 5.97 $\times 10^{-1}$       & 1.00 [20/20]        \\
                       &                              & 100         & [1.67, 2.36] $\times 10^{-0}$  &   1.36     & 5.97                        & 1.00 [20/20]        \\
    \bottomrule
\end{tabular}
\end{table*}

\begin{table*}
    \centering
    \caption{Summary of the Scaling Consistency Test of 20 independent runs with 5\,\% noise. Compared to the results in Table~\ref{table:SRT-noise01}, coefficient normalization decreases the success rate at $V_2$ and $V_3$, whereas gradient-weighted normalization retains the performance. }\label{table:SRT-noise05}
    \begin{tabular}{c|crcccc}
        \toprule
            dataset & normalization & scaling $\alpha$ & range & coeff. dist. & e.v. & success rate \\
        \midrule
        \multirow{10}{*}{$V_1$} & \multirow{5}{*}{coefficient}  & 0.01        & --                             & --         & --                          & 0.00 [00/20]       \\
                           &                              & 0.1         & --                             & --         & --                          & 0.00 [00/20]       \\
                           &                              & 1.0         & [2.28, 2.61] $\times 10^{-2}$  &  1.31      & 1.14 $\times 10^{-1}$       & 1.00 [20/20]       \\
                           &                              & 10          & [1.80, 1.91] $\times 10^{-0}$  &  1.41      & 4.13                        & 1.00 [20/20]       \\
                           &                              & 100         & [1.86, 1.91] $\times 10^{-0}$  &  1.20      & 3.34                        & 0.85 [17/20]       \\
                           \cmidrule{2-7}
                           & \multirow{5}{*}{gradient-weighted}    & 0.01        & --                             &            & --                          & 0.00 [00/20]        \\
                           &                              & 0.1         & [5.07, 5.66] $\times 10^{-3}$  &  1.22      & 2.61 $\times 10^{-2}$       & 1.00 [20/20]       \\
                           &                              & 1.0         & [5.07, 5.66] $\times 10^{-2}$  &  1.34      & 2.61 $\times 10^{-1}$       & 1.00 [20/20]       \\
                           &                              & 10          & [5.07, 5.66] $\times 10^{-1}$  &  1.37      & 2.61                        & 1.00 [20/20]       \\
                           &                              & 100         & [5.07, 5.66] $\times 10^{-0}$  &  1.40      & 2.61 $\times 10^{+1}$       & 1.00 [20/20]       \\
        \midrule
        \multirow{10}{*}{$V_2$} & \multirow{5}{*}{coefficient}  & 0.01        & --                             & --         & --                          & 0.00 [00/20]       \\
                           &                              & 0.1         & --                             & --         & --                          & 0.00 [00/20]       \\
                           &                              & 1.0         & --                             & --         & --                          & 0.00 [00/20]       \\
                           &                              & 10          & [3.64, 4.21] $\times 10^{-0}$  & 0.693      & 3.41                        & 1.00 [20/20]       \\
                           &                              & 100         & --                             &  --        & --                          & 0.00 [00/20]       \\
                           \cmidrule{2-7}
                           & \multirow{5}{*}{gradient-weighted}    & 0.01        & [4.11, 5.42] $\times 10^{-3}$  &  0.692     & 2.16 $\times 10^{-3}$       & 1.00 [20/20]        \\
                           &                              & 0.1         & [4.11, 5.42] $\times 10^{-2}$  &  0.667     & 2.16 $\times 10^{-2}$       & 1.00 [20/20]        \\
                           &                              & 1.0         & [4.11, 5.42] $\times 10^{-1}$  &  0.515     & 2.16 $\times 10^{-1}$       & 1.00 [20/20]        \\
                           &                              & 10          & [4.11, 5.42] $\times 10^{-0}$  &  0.565     & 2.16                        & 1.00 [20/20]        \\
                           &                              & 100         & [4.11, 5.42] $\times 10^{+1}$  &  0.692     & 2.16 $\times 10^{+1}$       & 1.00 [20/20]        \\
        \midrule
        \multirow{10}{*}{$V_3$} & \multirow{5}{*}{coefficient}  & 0.01        & --                             &  --        & --                          & 0.00 [00/20]       \\
                           &                              & 0.1         & [1.00, 1.00] $\times 10^{-6}$  &  0.416     & 2.10 $\times 10^{-6}$       & 0.30 [06/20]       \\
                           &                              & 1.0         & [0.90, 1.04] $\times 10^{-2}$  &   1.29     & 1.40                        & 0.95 [19/20]       \\
                           &                              & 10          & [1.32, 1.39] $\times 10^{-0}$  &   4.86     & 1.33                        & 1.00 [20/20]       \\
                           &                              & 100         & [1.67, 1.67] $\times 10^{-0}$  &   2.51     & 8.48 $\times 10^{-1}$       & 0.60 [12/20]       \\
                           \cmidrule{2-7}
                           & \multirow{5}{*}{gradient-weighted}    & 0.01        & [3.13, 3.68] $\times 10^{-4}$  &   1.38     & 2.37 $\times 10^{-3}$       & 1.00 [20/20]        \\
                           &                              & 0.1         & [3.13, 3.68] $\times 10^{-3}$  &   1.38     & 2.37 $\times 10^{-2}$       & 1.00 [20/20]        \\
                           &                              & 1.0         & [3.13, 3.68] $\times 10^{-2}$  &   1.34     & 2.37 $\times 10^{-1}$       & 1.00 [20/20]        \\
                           &                              & 10          & [3.13, 3.68] $\times 10^{-1}$  &   1.17     & 2.37                        & 1.00 [20/20]        \\
                           &                              & 100         & [3.13, 3.68] $\times 10^{-0}$  &   1.37     & 2.37 $\times 10^{+1}$       & 1.00 [20/20]        \\
        \bottomrule
    \end{tabular}
\end{table*}

\section{Conclusion}
We proposed gradient-weighted normalization for the basis computation of approximately vanishing ideals. Alhough the Euclidean norm of the gradient evaluated at a prescribed point set only defines a semi-norm, we prove that all quantities---(order ideal) terms, border terms, and border basis polynomials--- that are considered during the basis computation can be gradient-weighted normalized.
The introduction of gradient-weighted normalization is compatible with the existing analysis of approximate border bases and the computation algorithms. The time complexity does not change either.

The data-dependent nature of gradient-weighted normalization possesses important properties: stability against perturbation and invariance with respect to the scale of the point set.
In particular, through theory and numerical experiments, we highlighted the critical effect of the scaling of points on the success of approximate basis computation. We consider that the present study provides a new perspective and ingredients to analyze the border basis computation in the approximate setting, where perturbed points should be dealt with, and stable computation is required. 

Some methods that neither rely on eigenvalue problems nor on SVD solve simple quadratic programs, e.g., as least-squares problems. The implementation of gradient-weighted normalization in these methods is left to future work.

\section*{Acknowledgments}
We are grateful to Sebastian Pokutta, Elias Wirth, and the Zuse Institute Berlin for hosting us at the same time and, thus, giving us the opportunity to collaborate there offline. Hiroshi Kera would also like to thank Yuichi Ike for helpful discussions. This research was partially supported by JST PRESTO Grant Number JPMJPR24K4, JSPS KAKENHI Grant Number JP23KK0208, Mitsubishi Electric Information Technology R\&D Center, and the Chiba University IAAR Research Support Program and the Program for Forming Japan's Peak Research Universities (J-PEAKS).

\bibliographystyle{plain}

\clearpage

\section*{Appendices}
\appendix

\renewcommand{\thesection}{\Alph{section}}

\section{Comparison of border basis computation methods for perturbed points}\label{app:method-comparison}\label{sec:bb-computation-comparison}

Here, we compare border basis computation methods for perturbed points and justify the adoption of the ABM algorithm. We compare the ABM algorithm with two algorithms: the approximately vanishing ideal (AVI) algorithm~\citep{heldt2009approximate} and the stable order ideal~(SOI) algorithm~\citep{abbott2008stable}. The comparison of these three algorithms can also be found in \citep{limbeck2013computation}.

\paragraph{The AVI algorithm}\label{paragraph:bb-computation-comparison-avi}
The AVI algorithm is a seminal algorithm for the approximate computation of vanishing polynomials and computes approximate border bases efficiently. Several variants of this algorithm were developed (the AVI-family algorithms~\citep{limbeck2013computation}), including the ABM algorithm. The vanishing component analysis~\citep{livni2013vanishing} in machine learning was also inspired by this algorithm. However, the AVI algorithm has several caveats. We use the same notation as in Algorithm~\ref{alg:ABMGN} for the comparison.
\begin{itemize}
    \item As with the ABM algorithm, the AVI algorithm takes $\epsilon$ as an input. However, the polynomials in the computed basis are not necessarily $\epsilon$-approximately vanishing. The threshold $\epsilon$ and the upper bound $\delta$ of the extent of vanishing only share a soft connection. 
    \begin{align}
        \delta &= \qty(\epsilon\sqrt{|G|} + \tau |G|(|G| + |\mathcal{O}|) )\sqrt{|X|}\cdot \|X\|_{\infty}^{\degree{G}},
    \end{align}
    where $\|X\|_{\infty} = \max_{\mathbf{x}\in X}\ \|\mathbf{x}\|_{\infty}$ ($\|\cdot\|_{\infty}$ denotes the $L_0$ norm), and $\tau > 0$ is an additional hyperparameter to determine the precision of the algorithm.\footnote{The $\delta$ here is subtly different from the original one, which is $\tau = \epsilon \sqrt{|G|} + \tau \sqrt{|G|}(|G| + |\mathcal{O}|)\sqrt{|X|}$, in \citep{heldt2009approximate} because we fixed a minor error and did not use the assumption $X\subset [-1, 1]^n$.} Evidently, $\delta$ is impractically large and cannot be known before the computation. 
    \item The computed approximately vanishing polynomials are not optimal in the least squares sense. 
    \item Unlike the ABM algorithm, the AVI algorithm's output is not affected by the presence of multiple (nearly) identical points in $\XX$. However, in several applications, the multiplicity of points indicates their reliability, and thus, polynomials should be more fitted to them, whereas the outliers tend to have no multiplicity. This implies the sensitivity of the AVI algorithm to outliers. 
\end{itemize}
All these caveats originate from the (stabilized) reduced row echelon form calculation step (cf. Appendix~\ref{sec:avi+gwn}). Not only the practical caveats mentioned above but also the calculation of the reduced row echelon form make the algorithm too complex for theoretical analysis. We chose the ABM algorithm as the object of analysis because it can circumvent the calculation of the reduced row echelon form. Therefore, it is free from all the caveats. 

\paragraph{The SOI algorithm} The SOI algorithm takes a different approach from the AVI and ABM algorithms. Instead of an approximate border basis, it computes an order ideal that is \textit{stable} against perturbation of points. If a border basis is needed, then any algorithm that computes a border basis from an order ideal can be applied. Although the SOI algorithm can more directly control the stability of the border basis in theory, it instead requires several impractical assumptions. 
\begin{itemize}
    \item The SOI algorithm assumes that the maximal magnitude of noise is known a priori. However, several applications involving unbounded noise (e.g., additive Gaussian noise) are considered. 
    \item The SOI algorithm is sensitive to the choice of the value of the hyperparameter $\gamma$, which is related to the supremum of the residual error within a small ball of nonlinear terms of the Taylor expansion of a polynomial obtained during the algorithm execution. Thus, it is very difficult to set the value $\gamma$ properly. 
\end{itemize}
The SOI algorithm is also known for its high computation cost. For example, with 40 three-dimensional generic points, the SOI algorithm takes approximately 3,000 seconds, whereas the ABM algorithm takes 0.25 seconds~\citep{limbeck2013computation}. 

To summarize, although both the AVI and SOI algorithms are seminal, the ABM algorithm is superior to them in practice and also useful for analysis, motivating us to adopt it as an example.

\section{The AVI Algorithm with Gradient-Weighted Normalization}\label{sec:avi+gwn}

Our replacement of normalization can be integrated to various AVI-family algorithms because they are all based on solving eigenvalue problems or the singular value decomposition. Here, we provide an example of the AVI algorithm with gradient-weighted normalization. 

In the following, the original AVI algorithm\footnote{Strictly speaking, the AVI algorithm presented here is based on~\citep{heldt2009approximate}, and it does not give an border prebasis due to a minor error. The fixed version is presented in~\citep{limbeck2013computation}.} is modified at \textbf{S2} and \textbf{S3}, where a generalized eigenvalue problem instead of an SVD is solved as in Eq.~\eqref{eq:gep} of the ABM algorithm with gradient-weighted normalization. 
Because the AVI algorithms manipulate degree-$d$ border terms simultaneously, not only the generalized eigenvector of the smallest eigenvalue but also those of other eigenvalues not exceeding ${\epsilon}^2$ are calculated. Note that the polynomials in $G$ output by the original AVI algorithm are not necessarily $\epsilon$-approximately vanishing because of the calculation of the stabilized row echelon form at \textbf{S3} (see Section~\ref{sec:bb-computation-comparison}). 

Let $\XX\subset[-1, 1]^n$ be a non-empty, finite set of points, $\epsilon > \tau > 0$ error tolerances, and $\sigma$ a degree-compatible term ordering. As in the ABM algorithm, the AVI algorithm collects order terms and approximately vanishing polynomials from lower to higher degrees. The algorithm starts with lists $G = []$, $\mathcal{O} = [1]$, and degree parameter $d= 0$. The parameters $r=0$ and $s=1$ count the elements in $G$ and $\mathcal{O}$ respectively.
\begin{itemize}
    \item[\textbf{S1}] Increase $d$ by one and let $T_d = [b_1,\dots, b_{n_d}]$ the trial terms of degree $d$ in $\partial \mathcal{O}$ in $\sigma$-decreasing order. If $T_d$ is an empty set, the algorithm outputs $(G, \mathcal{O})$ and terminates. 
    
    \item[\textbf{S2}] Let $\mathcal{O}=[t_1,\dots,t_s]$. Let $M = \mqty(T_d(\XX) &\mathcal{O}(\XX)) \in \mathbb{R}^{m\times (n_d+s)}$. 
    \textbf{Solve the simultaneous generalized eigenvalue problem} 
    \begin{align}
        M^{\intercal}MV = D^2V\Lambda,
    \end{align}
    where $D = \text{diag}\bigl(\gwnorm{b_1}, \ldots, \gwnorm{b_{n_d}}, \gwnorm{t_1}, \ldots, \gwnorm{t_s}\bigr)$, while $\Lambda$ the diagonal matrix of generalized eigenvalues and $V$ the matrix of generalized eigenvectors. 
    
    Let $B$ be a matrix whose columns consist of the generalized eigenvectors corresponding to the generalized eigenvalues less than or equal to $\epsilon^2$. If $B$ is the empty matrix, i.e., if there is no generalized eigenvalue of value at most $\epsilon^2$, then go to step S1.
    
    \item[\textbf{S3}] Compute the stabilized reduced row echelon form~\citep{heldt2009approximate} of $B^{\intercal}$ with respect to the lower tolerance $\tau$ (\textbf{the coefficient normalization in the postprocessing phase is replaced with gradient-weighted normalization}). The result is a matrix $C \in\mathbb{R}^{k\times (s+n_d)}$ such that $c_{ij} = 0$ for $j \le \nu(i)$, where $\nu(i)$ denotes the column index of the pivot element in the $i$-th row of $C$.
    
    \item[\textbf{S4}] For each index $j \in \{1, \ldots, n_d\}$ whose column possess a pivot, say in row $i$, append the polynomial
    \begin{align}
        g_{r+1}=c_{ij}b_{j} + \sum_{k = j+1}^{n_d} c_{ik}b_{k} + \sum_{k=1}^{s} c_{i,n_d+k}t_{k}
    \end{align}
    to the list $G$ and increase $r$ by one.
    
    \item[\textbf{S5}] For each index $j=n_d, n_d-1, \ldots, 1$ such that the $j$-th column of $C$ is pivot-free, append the term $t_{s+1}=b_j$ as a new first element to $\mathcal{O}$ and increase $s$ by one. 
    
    \item[\textbf{S6}] Continue with step \textbf{S2}. 
\end{itemize}

\section{Singular Value Decomposition as Alternative for the Generalized Eigenvalue Problem}\label{sec:ABMGN-svd}

Some programming languages do not support a solver of the generalized eigenvalue problem. In this case, one can use the singular value decomposition (SVD) instead. For example, the step \textbf{F1} of ABM+GWN algorithm can be implemented by SVD as follows.

\begin{itemize}
    \item[\textbf{F1$^\prime$}] Consider the reduced set of terms $\mathcal{O}_{-} = \mathcal{O} \setminus \{1\} = \{t_2, \ldots, t_s\}$. Then let $M_{-} = \mqty(\mathcal{O}_{-}(X) & b(X))$ and $D_{-} = \mathrm{diag}(\gwnorm{t_2},\dots, \gwnorm{t_s}, \gwnorm{b})$, which is positive definite.
    Compute the SVD of $M_{-}D_{-}^{-1}$, 
    \begin{align}
        M_{-}D_{-}^{-1} &= U\Sigma V^{\intercal}\ ,\qquad 
        U=\begin{pmatrix}
            \vert &  &\vert \\
            \mathbf{u}_1 & \cdots & \mathbf{u}_m\\
            \vert & &\vert \\
        \end{pmatrix}
        \ ,\quad
        V=\begin{pmatrix}
            \vert & &\vert \\
            \mathbf{v}_1 & \cdots & \mathbf{v}_s\\
            \vert &  &\vert \\
        \end{pmatrix}
        \ ,\text{\ and\ }
        \Sigma=\begin{pmatrix}
            \sigma_1 &  & 0\\
            & \ddots &\\
            0 & & \sigma_s\\
            0 & \dots & 0\\
            \vdots & & \vdots\\
            0 & \dots & 0
            \end{pmatrix}
    \end{align}
    with the orthogonal matrices $U$ and $V$ and singular values $\sigma_1\ge \dots\ge \sigma_s$.
    Let $\sigma_{\min}=\sigma_s$ be the smallest singular value and $\mathbf{v}_s$ its right singular vector. 
    Let $\hat{\mathbf{v}}_{\min} = D_{-}{\mathbf{v}_s}$ and let $v_0$ be the mean value of the components of $\sigma_{\min}\hat{\mathbf{v}}_{\min}$. Then, these results produce the minimal solution of the non-reduced generalized eigenvalue problem,
    \begin{align}
        \mathbf{v}_{\min} = \text{sign}(v_b)\cdot\mqty(-v_0 & \hat{\mathbf{v}}_{\min}^{\intercal} )^{\intercal}\ ,\qquad
        \lambda_{\min} = \enorm{M\mathbf{v}_{\min}}^2\ ,
    \end{align}
    where $\text{sign}(v_b) \in \{-1, 1\}$ denotes the sign of the border term coefficient.
\end{itemize}

\begin{remark}
The removal of $t_1 = 1$ makes $D_{-}$ invertible. This allows us to rewrite the generalized eigenvalue problem $M_{-}^{\intercal}M_{-}\mathbf{w}=\sigma D^2_{-}\mathbf{w}$ as the eigenvalue problem $D^{-1}_{-}M_{-}^{\intercal}M_{-}D^{-1}_{-}(D_{-}\mathbf{w})=\sigma(D_{-}\mathbf{w})$, which coincides with the SVD of $M_{-}D^{-1}_{-}$.
\end{remark}
The smallest generalized eigenvalue $\lambda_{\min}$ can be computed as follows. Let $\mathbf{1}=(1,\dots,1)^\intercal$ be the vector of length $m=|\XX|$. Then,
\begin{align*}
    \lambda_{\min} 
    &= \|M\mathbf{v}_{\min}\|^2 =\mathbf{v}_{\min}^\intercal M^\intercal M\mathbf{v}_{\min}\\ 
    &= \mqty(-v_0 & \hat{\mathbf{v}}_{\min}^{\intercal} ) 
    \begin{pmatrix}
        \mathbf{1}^\intercal \\
        M_{-}^\intercal
    \end{pmatrix}
    \begin{pmatrix}
        \mathbf{1}& M_{-}
    \end{pmatrix}
    \mqty( -v_0 \\ \hat{\mathbf{v}}_{\min})=
    \mqty(-v_0 & \hat{\mathbf{v}}_{\min}^{\intercal} ) 
    \mqty( |\XX| & \mathbf{1}^\intercal M_{-}\\ M_{-}^{\intercal}\mathbf{1} & M_{-}^{\intercal}M_{-} ) \mqty( -v_0 \\ \hat{\mathbf{v}}_{\min}) \\
    & = v_0^2 \cdot |\XX| - 2v_0\mathbf{1}^{\intercal}M_{-}\hat{\mathbf{v}}_{\min}+\hat{\mathbf{v}}_{\min}^\intercal M_{-}^\intercal M_{-}\hat{\mathbf{v}}_{\min}\\
    & = v_0^2 \cdot |\XX| - 2v_0\mathbf{1}^{\intercal}M_{-}\hat{\mathbf{v}}_{\min}+\sigma_{\min}^2\ . 
\end{align*}

For implementation, we recommend solving the generalized eigenvalue problem directly instead of reducing it and computing the SVD because \textbf{S2$^\prime$} involves more steps, which can accumulate numerical errors.

\end{document}